\documentclass[reqno]{amsart}

\newtheorem{theorem}{Theorem}[section]

\newtheorem{lemma}[theorem]{Lemma}

\newtheorem{definition}[theorem]{Definition}
\newtheorem{remark}[theorem]{Remark}

\newenvironment{acknowledgement}{\emph{Acknowledgement.}}

\numberwithin{equation}{section}

\begin{document}

\title[Multiscale Analysis and Localization]
{Multiscale Analysis and Localization of random 
 operators}

\author{Abel Klein}
\address{University of California, Irvine,
Department of Mathematics,
Irvine, CA 92697-3875,  USA}
 \email{aklein@uci.edu}

\thanks{2000 \emph{Mathematics Subject Classification.} 
Primary 82B44; Secondary  47B80, 60H25}
\thanks{A.K. was supported in part   by NSF Grant
 DMS-0200710}

\begin{abstract}  {A discussion of the method of multiscale analysis in the study of
 localization of random  operators based on lectures given at   
\emph{Random Schr\"odinger operators: 
 methods, results, and perspectives},  
\'Etats de la recherche, Universit\'e Paris 13, June 2002}
\end{abstract}

\maketitle

\tableofcontents

\section{Introduction}

In his seminal 1958  article \cite{And}, Anderson argued that for
a simple Schr\"odinger operator in a  disordered medium,``at sufficiently low
 densities transport does not take place; the exact wave functions are localized in a small
 region of space."   This phenomenon,  known as Anderson
 localization, originally  studied in the context
of quantum mechanical electrons in random media 
(e.g., \cite{T}), was later found relevant also in the context   of  classical waves 
in random media (e.g., \cite{And2,Ma,Jo1,Jo2}), where it 
was  observed in light waves in an experiment conducted by Wiersma et al 
 \cite{WBLR}.

Anderson localization was initially given a spectral
  interpretation:  pure point
spectrum with exponentially decaying eigenstates (exponential localization).
But the intuitive physical notion of 
localization has also a dynamical interpretation: 
the moments of a
 wave packet, initially localized both in
space and in energy,  should remain 
uniformly bounded under  time evolution.
 (Dynamical localization implies
pure point spectrum, but the converse is not true.) 
Although exponential localization has sometimes been called Anderson
 localization, we will use Anderson localization in a broader sense,
 since it can be argued  the circle
 of ideas
regarding localization, originating from 
\cite{And}, include the physical notion of dynamical localization.

Localization for random operators was first established in the celebrated
paper by  Gol'dsheid,  Molchanov and Pastur \cite{GMP} for a
 one dimensional continuous random Schr\"odinger operator. 
Their method was extended to other
 one and quasi-one (the strip) dimensional  
random Schr\"odinger operators \cite{KS,Car,Lac}.  But the
 multi-dimensional case required new methods.

 The method  with the wider applicability has been the
 multiscale analysis,  a
technique initially  developed by Fr\"ohlich and Spencer  \cite{FS}
and Fr\"ohlich, Martinelli, Spencer and Scoppolla
 \cite{FMSS}, and simplified by von Dreifus \cite{VD} and  von Dreifus
and Klein
\cite{VDK}. 
 (For the multiscale analysis 
per se, see
also
 \cite{HM,Sp,VDK2,K,Gr,Klo5,CH1,FK1,KSS,KSS2,Kr,St,GK1,GK4}, 
for applications see also
\cite{CKM,KlMP,KLS,Kl0,Kl,FK3,FK4,CH2,FK2,W1,BCH1,BCH2,SVW,CHT,K3,DBG,
FLM,Kl4,Z,DSS1,U,KK2,GK3,GK5,GK6}.)
Although it originally only gave exponential localization 
\cite{FMSS,DelyLS,SW,VDK,CH1},  it was later shown 
to also yield dynamical localization  by
 Germinet and  De Bi\`evre \cite{GDB},   strong dynamical localization
 for moments up to some finite order  
by Damanik and Stollman \cite{DS}, and 
 strong dynamical localization (up to all orders) in the Hilbert-Schmidt
norm by Germinet and Klein \cite{GK1}. The latest version of the multiscale analysis,
  the bootstrap multiscale analysis
of Germinet and Klein \cite{GK1}, 
 built out of four different
multiscale analyses,  
 yields exponential localization, semi-uniformly localized eigenfunctions
(SULE), and sub-exponential
decay of  the expectation of the kernel of the evolution operator.

The other successful method for proving localization in the multi-dimensional case
is the fractional moment method  introduced by Aizenman and 
Molchanov \cite{AM,A,ASFH}, which has just been extended to the continuum 
by Aizenman et al  \cite{AENSS}.  It yields exponential
decay for the  expectation of the kernel of the evolution operator, but it requires
that the conditional expectation of certain random variables
have bounded densities.

In these lectures we discuss the method of multiscale analysis in the study of
 localization of random  operators.    A random medium will be modeled by a ergodic  random
self-adjoint operator.   In Section~\ref{secrandomop} we discuss the most important
random operators: random Schr\"odinger operators, random Landau Hamiltonians,
and random classical wave operators (Maxwell, acoustic, elastic).
In Section~\ref{secloc} we discuss several
 definitions of localization from both the spectral and dynamical  points of view.
In Section~\ref{secreq} we describe the properties of random operators required by
 the multiscale analysis.  In Section~\ref{secMSA} we state  and discuss
the bootstrap multiscale analysis
plus the four multiscale analyses used in its proof. In Section~\ref{secMASAloc} we 
 prove exponential and dynamical localization from the multiscale analysis.  
In Section~\ref{proofMSA} we show how to 
perform a multiscale analysis; we give a complete proof of  the Dreifus-Klein multiscale analysis
in the continuum.

These lectures were written in 2002.  Since then Bourgain and Kenig \cite{BK}  proved localization in the continuous Anderson-Bernoulli model, using a multiscale analysis. The Wegner estimate is established   in the multiscale analysis using  ``free sites" and a new quantitative version of unique  continuation which gives  a lower bound on eigenfunctions. 
Since their Wegner estimate  has weak probability estimates and  the underlying random variables are discrete, they also introduced a new method to prove Anderson localization from estimates on the finite-volume resolvents given by a single-energy multiscale analysis.  The new method does not use spectral averaging as in \cite{CH1,DelyLS,SW}, which requires random variables with bounded densities.  It is also not an energy-interval multiscale analysis as in \cite{VDK,FMSS}, which requires better probability estimates.
Subsequently, Germinet, Hislop and Klein
\cite{GHK1, GHK2, GHK3}  proved localization for Schr\"odinger operators with  Poisson random potential, using   a multiscale analysis that 
exploits  the   probabilistic properties of Poisson point processes to
 control  the randomness of the configurations, and at the same time allows  the  use of the  new ideas introduced by  Bourgain and Kenig.

\section{Random operators} \label{secrandomop}

Quantum and classical waves in   random media 
are modeled by  random
self-adjoint operators  on either 
$\mathrm{L}^2(\mathbb{R}^d,{\rm d}x; \mathbb{C}^n)$
or $\ell^2(\mathbb{Z}^d; \mathbb{C}^n)$.    Examples include:

\begin{itemize}
\item Random Schr\"odinger operators:

\begin{itemize}

\item The Anderson model:   
\begin{equation}
 H_\omega = - \Delta + V_\omega \;\;
 \mbox{on}\;\; \ell^2(\mathbb{Z}^d) \, ,
\end{equation}
where $\Delta$ is the finite difference Laplacian and
$\{V_\omega(x); \ {x \in  \mathbb{Z}^d}\}$  are 
independent identically distributed bounded  random
variables.  (E.g., \cite{KS,FS,Lac,FMSS,CKM,MS,KlMP,CFKS,VDK,Sp,KLS,K,Gr,AM,A,
FK3,K1,K2,
SVW,ASFH,W2,Kl3}.)
\medskip

\item Anderson Hamiltonians on the continuum:

\begin{equation} \label{schr}
H_\omega = - \Delta + V_{\mathrm{per}} +  V_\omega \; \; \; \mathrm{on} \; \; \;   
\mathrm{L}^2(\mathbb{R}^d, {\mathrm{d}}x), 
\end{equation}
where $\Delta$ is the  Laplacian operator,
$V_{\mathrm{per}}$ is a periodic potential
 (by rescaling  we take the period to be one) of the form
$V_{\mathrm{per}}=
V_{\mathrm{per}}^{(1)} + V_{\mathrm{per}}^{(2)} $,
with $V_{\mathrm{per}}^{(i)} $, $i=1,2$, periodic with period one,
$0\le V_{\mathrm{per}}^{(1)} \in 
 \mathrm{L}_{\mathrm{loc}}^1(\mathbb{R}^d, {\mathrm{d}}x)$, 
  $ V_{\mathrm{per}}^{(2)}$  relatively form-bounded with respect to
$-\Delta$ with relative bound $<1$,  and
$V_\omega$  a random potential of the form
\begin{equation} \label{pot}
V_\omega(x)=\sum_{i\in  \frac1{q}\mathbb{Z}^d}
\omega_i \, u(x-i),
\end{equation}
where $q \in \mathbb{N}$,
$\omega =\{\omega_i; \ {i\in  \frac1{q}\mathbb{Z}^d}\}$  are 
independent identically distributed bounded  random
variables,  $u$ is  a real valued measurable function with 
compact support, $u \in \mathrm{L}^{p}(\mathbb{R}^d, {\mathrm{d}}x) $
  with 
$ {p} >   \frac d 2$ if
 $d\ge 2$ and
 ${p}=2$ if $d=1$. (E.g., \cite{HM,Klo5,Kl0,CH1,Kl,BCH1,KSS,KSS2,GDB,St,GK1,DS,DSS1,
Kl4,Z,GK3,GK4,GK5,GK6,AENSS}.)
\end{itemize}\bigskip

\item Random Landau Hamiltonians:
\begin{eqnarray}
H_\omega = H_0 +  V_\omega \; \; \; \mathrm{on} \; \; \;  
\mathrm{L}^2(\mathbb{R}^2, {\mathrm{d}}x), 
\end{eqnarray}
where $H_0 =  (-i\nabla -A)^2$, $A=\frac B2 (x_2,-x_1)$
with $B>0$, and the random potential 
$V_\omega$ is as 
in \eqref{pot} with $q=1$ and $u(x)$ bounded. 
 (See \cite{CH2,W1,BCH2,GK4}.)\bigskip

\item Random classical wave operators: 
\begin{itemize}

\item Maxwell operators in random media:
\begin{equation}
H_\omega =
\frac 1{\sqrt{\mu_\omega (x)}}\nabla{\times}
 \frac 1{\varepsilon_\omega  (x)}\nabla{\times}
\frac 1{\sqrt{\mu_\omega (x)}}  \; \; \; \mathrm{on} \; \; \;   
\mathrm{L}^{2}\left( \mathbb{R}^{3},dx;\mathbb{C}^3\right)
\end{equation}
 where $\nabla {\times}$ is the operator given by the curl,
 $\varepsilon_\omega(x) $   is 
the random dielectric constant and $\mu_\omega(x) $   is the 
random magnetic permeability.  We take  
\begin{eqnarray} \label{varepsilon}
\varepsilon _{\omega }(x)&=&\varepsilon _0(x)\gamma _{\omega }(x)
\;\;,\mbox{with}\;\;
\gamma _{\omega }(x)=1+\sum_{i\in \frac 1q \mathbb{Z}^3}\omega _iu(x-i), \\
\mu _{\omega }(x)&=&\mu _0(x)\beta _{\omega }(x)
\;\;,\mbox{with}\;\;
\beta _{\omega }(x)=1+\sum_{i\in \frac 1q \mathbb{Z}^3}\omega _iv(x-i),\label{mu}
\end{eqnarray}
where $q \in \mathbb{N}$,
$\omega =\{\omega_i; \ {i\in  \frac1{q}\mathbb{Z}^d}\}$  are 
independent identically distributed bounded  random
variables
  taking values in the
interval $[-1,1]$, $\varepsilon _0(x)$ and   $\mu _0(x)$ are periodic
measurable functions (by rescaling  we take the period to be one), such that
$
0<\varepsilon _{-}\leq \varepsilon \left( x\right) \le\varepsilon
_{+}<\infty$ and 
$0<\mu _{-}\leq \mu \left( x\right) \le\mu_{+}<\infty$
for some constants $\varepsilon _{\pm}$ and $\mu _{\pm}$, 
$u(x)$ and
$v(x)$ are
nonnegative measurable real valued functions with compact support,
 such that
\begin{eqnarray}
0 \le U_- \le U(x)\equiv \sum_{{i\in \frac 1q \mathbb{Z}^3}}u_i(x)&\le& U_+<\infty, 
 \label{muM}\\
0 \le V_-\le V(x)\equiv \sum_{{i\in \frac 1q \mathbb{Z}^3}}v_i(x)&\le& V_+<\infty, 
 \label{muM2}
\end{eqnarray}
for some constants $U_\pm$ and $V_\pm$, with $ U_- + V_- >0$
and ${\max\{U_+, V_+\}}<1 $.
(See \cite{FK4,FK2,K3,CHT,KK1,KK2}.)\medskip

\item Acoustic operators in random media:
\begin{equation}
H_\omega = \frac{1}{\sqrt{\kappa_\omega (x)}}\nabla^* \frac{1}{
\rho_\omega (x)}\nabla \frac{1}{\sqrt{\kappa_\omega (x)}} 
\; \; \; \mathrm{on} \; \; \;   
\mathrm{L}^2(\mathbb{R}^d, {\mathrm{d}}x), \label{acoustic}
\end{equation} 
  where  $\nabla$ is the gradient operator, and
  the  random compressibility  $\kappa_\omega (x)$ and   the random
mass density $\varrho_\omega (x)$ are of the same form as  
$\varepsilon _{\omega }(x)$ and
$\mu _{\omega }(x)$ in \eqref{varepsilon} and \eqref{mu}.  
(See \cite{FK4,FK1,CHT,KK1,KK2}).\medskip

\item Elastic operators in random media:
\begin{eqnarray}\label{elastic}
\lefteqn{H_\omega =}\\
 && \frac{1}{\sqrt{\rho_\omega (x)}}
 \left\{
\nabla \left(\lambda_\omega(x) +2 \mu_\omega(x)\right)\nabla^*
 +\nabla \times \mu_\omega(x) \nabla\times \right\}
\frac{1}{\sqrt{\rho_\omega (x)}} 
 \nonumber
\end{eqnarray} 
on 
$\mathrm{L}^{2}\left( \mathbb{R}^{3},dx;\mathbb{C}^3\right)$, 
 where   the mass  density $\rho_\omega(x) $, and the 
 Lam\'e moduli
 $\lambda_\omega(x)$ and $\mu_\omega(x)$ are
 of the same form as  
$\varepsilon _{\omega }(x)$ and
$\mu _{\omega }(x)$ in \eqref{varepsilon} and \eqref{mu}. 
 (See \cite{KK1,KK2}).
\end{itemize}

\end{itemize}

In all these examples the random operator
 $H_\omega$ is a $\mathbb{Z}^d$-ergodic  random
self-adjoint operator $H_\omega$ on a Hilbert space
$\mathcal{H}$, where $\omega$ belongs to a
set
$\Omega$ with a probability measure $\mathbb{P}$
and expectation $\mathbb{E}$, and either 
$\mathcal{H}=\mathrm{L}^2(\mathbb{R}^d,{\rm d}x; \mathbb{C}^n)$
(``on the continuum")
or $\mathcal{H}=\ell^2(\mathbb{Z}^d; \mathbb{C}^n)$
(``on the lattice"). 
They all
satisfy the following definition.

\begin{definition}  An ergodic random operator is a $\mathbb{Z}^d$-ergodic
 measurable map
 $H_\omega$
from a probability space $(\Omega,\mathcal{F},\mathbb{P})$
(with  expectation $\mathbb{E}$) to self-adjoint operators on
 either 
$\mathrm{L}^2(\mathbb{R}^d,{\rm d}x; \mathbb{C}^n)$
or $\ell^2(\mathbb{Z}^d; \mathbb{C}^n)$.   
\end{definition}

By measurability of $H_\omega$ we mean that
 the mappings $\omega \to f(H_\omega)$ are weakly
(and hence strongly)
measurable for all bounded Borel  measurable functions $f$ on $\mathbb{R}$.
 (See \cite{KM}, \cite[Section V.1]{CL} for more details.)
Random operators may be defined without any 
ergodicity requirement, ergodicity being an extra
 requirement, but since we will  be dealing only
with $\mathbb{Z}^d$-ergodic random operators, we included 
it in the definition for convenience.  We recall that 
 $H_\omega$ is  $\mathbb{Z}^d$-ergodic if 
there exists  a group representation of  $\mathbb{Z}^d$ by an ergodic family
$\{ \tau_y; \ y \in \mathbb{Z}^d\}$ of measure preserving
 transformations
on $ (\Omega, \mathcal{F}, \mathbb{P})$ 
 such that 
\begin{equation} \label{cov}
U(y) H_\omega U(y)^*=  H_{\tau_y(\omega)}\;\; 
\mbox{for  all $ y \in \mathbb{Z}^d$},
\end{equation}
where $U(y)$ is the unitary operator given
by translation: $(U(y)f)(x) = f(x-y)$.  (Note that for Landau Hamiltonians
translations are replaced
by magnetic translations.)

An important  consequence of ergodicity is 
 that there exists a nonrandom set $\Sigma $
such that $\sigma (H_\omega)=\Sigma $ with probability one, where 
$\sigma (A)$ denotes the spectrum of the operator $A$. In addition, the
decomposition of $\sigma (H_\omega)$ into pure point spectrum
$\sigma_{pp} (H_\omega)$,
absolutely continuous spectrum  $\sigma_{ac} (H_\omega)$,
 and singular continuous spectrum $\sigma_{sc} (H_\omega)$  is also
independent of the choice of $\omega $ with probability one, i.e.,
there are  nonrandom sets $\Sigma_{pp}$, $\Sigma_{ac}$ and
 $\Sigma_{sc}$, such that  $\sigma_{pp} (H_\omega)=\Sigma_{pp} $,
 $\sigma_{ac} (H_\omega)=\Sigma_{ac} $, and
 $\sigma_{sc} (H_\omega)=\Sigma_{sc} $ with
probability one. (See
 \cite{Pa,KS,KM,PF,CL,CFKS}.)

\section{Spectral and dynamical localization} \label{secloc}

Localization can be interpreted from either 
 the spectral or the dynamical point of views.  We give selected 
 definitions from each point of view.

By  $\chi_B$ we denote the characteristic  function of the
set $B \subset  \mathbb{R}^d$ (or $\mathbb{Z}^d$).
By  $\chi_{x}$ we denote the characteristic  function of the
 cube of side $1$ centered at
 $x \in \mathbb{Z}^d$. We write
$\langle x \rangle= \sqrt{1+|x|^2}$.  The spectral projection
of $H_\omega$ is denoted by $E_\omega(\cdot )$. 
 The Hilbert-Schmidt norm of an operator $A$  is written as
$\|A\|_2$.

\begin{definition}\label{DLdef} 
Let $H_\omega$ be an ergodic random operator and   ${\mathcal{{{I}}}}$
an open interval.  Then
\begin{description}
\item[(i)\!\!]  $H_\omega$ exhibits spectral
localization (SL) in  ${\mathcal{{{I}}}}$
 if it has pure point spectrum in
 ${\mathcal{I}}$, i.e., $\Sigma \cap {\mathcal{I}}= 
\Sigma_{pp} \cap {\mathcal{I}}\not=\emptyset$ and
 $\Sigma_{ac} \cap {\mathcal{I}}=\Sigma_{sc} \cap {\mathcal{I}}
=\emptyset$. 

\item[(ii)\!\!]  $H_\omega$ exhibits exponential
localization (EL) in  ${\mathcal{{{I}}}}$ if it exhibits  spectral
localization in  ${\mathcal{{{I}}}}$ and
 for $\mathbb{P}$-almost every $\omega$ the eigenfunctions of
$H_\omega$ with eigenvalue in ${\mathcal{I}}$ decay 
exponentially in the $L^2$-sense.  (A function $\psi$ 
 decays exponentially   in the $L^2$-sense if
 $\|\chi_x\psi\|$ decays exponentially, i.e.,  $\|\chi_x\psi\|\le C \mathrm{e}^{-m|x|}$ with  $C$ and $m>0$ constants.)

\item[(iii)\!\!] 
$H_\omega$ exhibits 
dynamical localization (DL) in ${\mathcal{I}}$ if 
$\Sigma \cap {\mathcal{I}}\not=\emptyset$ and,
 for $\mathbb{P}$-almost every $\omega$,
 each compact interval $I \subset{\mathcal{I}}$,
and $\psi \in \mathcal{H}$ with compact support,
  we have
\begin{equation} \label{sdlocintro66}
 \sup_{t \in \mathbb{R}}
 \left\| {\langle} x {\rangle}^{\frac n2} E_\omega(I) 
{\mathrm{e}^{-i t H_\omega }}\psi
 \right\| <\infty \;\;\; \mbox{for all $n \ge 0\,$}\, .
\end{equation}

\item[(iv)\!\!] 
$H_\omega$ exhibits strong 
dynamical localization (SDL) in ${\mathcal{I}}$ if 
$\Sigma \cap {\mathcal{I}}\not=\emptyset$ and
 for 
 each compact interval $I \subset{\mathcal{I}}$
and $\psi \in \mathcal{H}$ with compact support,
  we have
\begin{equation} \label{sdlocintro667}
\mathbb{E} \left\{ \sup_{t \in \mathbb{R}}
 \left\| {\langle} x {\rangle}^{\frac n2} E_\omega(I)
{\mathrm{e}^{-i t H_\omega }}\psi
 \right\|^2 \right\}<\infty \;\;\; \mbox{for all $n \ge 0\,$}\, .
\end{equation}

\item[(v)\!\!] 
$H_\omega$ exhibits strong
HS-dynamical localization (SHSDL) in ${\mathcal{I}}$ if 
$\Sigma \cap {\mathcal{I}}\not=\emptyset$ and
for
 each compact interval $I \subset{\mathcal{I}}$
and bounded Borel set $B$ we have
\begin{equation} \label{sdlocintro}
\mathbb{E} \left\{ \sup_{t \in \mathbb{R}}
 \left\| {\langle} x {\rangle}^{\frac n2} E_\omega(I)
{\mathrm{e}^{-i t H_\omega }}\chi_B 
 \right\|_2^2\right\}<\infty \;\;\; \mbox{for all $n \ge 0\,$}\, .
\end{equation}

\item[(vi)\!\!] 
$H_\omega$ exhibits strong full
HS-dynamical localization (SFHSDL) in ${\mathcal{I}}$ if 
$\Sigma \cap {\mathcal{I}}\not=\emptyset$ and
for
 each compact interval $I \subset{\mathcal{I}}$
and bounded Borel set $B $
  we have
\begin{equation} \label{sdlocintro7}
\mathbb{E} \left\{ \sup_{|\!|\!|f|\!|\!|\leq 1}
 \left\| {\langle} x {\rangle}^{\frac n2} E_\omega(I)   
{f( H_\omega) }\chi_B 
 \right\|_2^2\right\}<\infty \;\;\; \mbox{for all $n \ge 0\,$}\, ,
\end{equation}
the supremum being taken
over all  Borel functions $f$ of a real variable, with 
$\!|\!|\!| f |\!|\!|=\sup_{t \in \mathbb{R}} |f(t)|$

\item[(vii)\!\!] 
$H_\omega$ exhibits strong sub-exponential
HS-kernel decay (SSEHSKD) in ${\mathcal{I}}$ if 
$\Sigma \cap {\mathcal{I}}\not=\emptyset$ and
for each compact interval $I \subset{\mathcal{I}}$
and $0<\zeta<1$ there is a finite constant $ C_{I,\zeta}$
 such that 
\begin{equation} \label{sdlocintro9}
\mathbb{E} \left\{ \sup_{|\!|\!|f|\!|\!|\leq 1}
 \left\|\chi_x E_\omega(I)
{f( H_\omega) }\chi_y 
 \right\|_2^2\right\} \leq  C_{I,\zeta} \,\mathrm{e}^{-|x-y|^\zeta}\, ,
\end{equation}
for all $x,y \in \mathbb{Z}^d$, the supremum being taken
over all  Borel functions $f$ of a real variable, with 
$\!|\!|\!| f |\!|\!|=\sup_{t \in \mathbb{R}} |f(t)|$.

\end{description}
\end{definition}

\begin{definition} 
Let $H_\omega$ be an ergodic random operator.
The spectral localization  region $\Sigma_{\mathrm{SL}}$,
 exponential localization  region $\Sigma_{\mathrm{EL}}$,
dynamical localization  region $\Sigma_{\mathrm{DL}}$,
strong dynamical localization  region  $\Sigma_{\mathrm{SDL}}$,
strong HS-dynamical localization  region $\Sigma_{\mathrm{SHSDL}}$,
strong full  HS-dynamical localization  region 
$\Sigma_{\mathrm{SFHSDL}}$,
strong sub-exponential HS-kernel decay region 
$\Sigma_{\mathrm{SSEHSKD}}$,
for the random  operator
 $H_\omega$, are defined as the set of $E \in\Sigma$ for
 which there exists some
open interval ${\mathcal{I}} \ni E$ such that $H_\omega$ exhibits
 spectral localization,  exponential localization,
dynamical localization, strong dynamical localization,
strong HS-dynamical localization,
strong full  HS-dynamical localization  region,
strong sub-exponential HS-kernel decay, respectively,
 in ${\mathcal{I}}$.
\end{definition}

\begin{remark}  Note that 
\begin{equation}
\Sigma_{\mathrm{SSEHSKD}} \subset \Sigma_{\mathrm{SFHSDL}}
 \subset \Sigma_{\mathrm{SHSDL}} \subset \Sigma_{\mathrm{SDL}}
\subset \Sigma_{\mathrm{DL}} \subset \Sigma_{\mathrm{SL}} \, .
\end{equation}
 That
$\Sigma_{\mathrm{SSEHSKD}} \subset \Sigma_{\mathrm{SFHSDL}}$
 is a simple calculation
(see \cite[Proof of Corollary 3.10]{GK1});  that
$ \Sigma_{\mathrm{SFHSDL}}
 \subset \Sigma_{\mathrm{SHSDL}} \subset \Sigma_{\mathrm{SDL}}
\subset \Sigma_{\mathrm{DL}} $ is obvious;  that
$ \Sigma_{\mathrm{DL}} \subset \Sigma_{\mathrm{SL}}$
 follows from  the RAGE Theorem (e.g., the argument in 
\cite[Theorem 9.21]{CFKS}). But dynamical localization is actually a
strictly stronger notion than pure point spectrum, since the  latter can take
place whereas a quasi-ballistic motion is observed \cite{DRJLS}.
\end{remark}

For an ergodic random operator with suitable properties, spelled out in the next section,
the original multiscale analyses showed that  decay of the resolvent in
 a finite, but large enough, volume with high probability (the ``starting
hypothesis" for the multiscale analysis))  gave a sufficient condition for
$E \in  \Sigma_{\mathrm{SL}} $ \cite{FS,FMSS,VD,VDK}. Later that
 condition was shown to be sufficient for $E \in \Sigma_{\mathrm{DL}}$
\cite{GDB}, $E \in \Sigma_{\mathrm{SDL}}$ \cite{DS} (more predisely,
they show that
\eqref{sdlocintro} holds with the operator
norm substituted for the Hilbert-Schmidt norm and  $n\le n_0 $ for
 some $n_0 < \infty$), and finally
$E \in\Sigma_{\mathrm{SSEHSKD}} $ \cite{GK1}.  Moreover, 
the converse was found to be true:
 $E \in\Sigma_{\mathrm{SHSDL}}$
implies the starting hypothesis of the
multiscale analysis \cite{GK3}.

\begin{remark} The multiscale analysis region $\Sigma_{\mathrm{MSA}}$
is given in Definition~\ref{sigmaSMA} as the region where the conclusions of the
 multiscale analysis hold.  If the ergodic random operator satisfies the requirements of the 
multiscale analysis in an open interval $\mathcal{I}$, it will be shown in
Theorem~\ref{tBMSA2} that $\Sigma_{\mathrm{MSA}} \cap \mathcal{I} \subset
\Sigma_{\mathrm{EL}}  \cap \Sigma_{\mathrm{SSEHSKD}}
 \cap \mathcal{I} $.  If in addition we have property \eqref{KK12}  and 
the kernel decay estimates of \cite{GK2} hold uniformly for $\mathbb{P}$-a.e. $\omega$
(both requirements are usually satisfied),    
then it is proven in \cite{GK3} that
\begin{equation}\label{spregion}
\Sigma_{\mathrm{MSA}} \cap \mathcal{I} = \Sigma_{\mathrm{SSEHSKD}}
 \cap \mathcal{I}=  \Sigma_{\mathrm{SHSDL}}
 \cap \mathcal{I}\, .
\end{equation}
Moreover, in \cite{GKsudec} it is shown that the spectral region  in  (\ref{spregion}) has characterizations  by the  decay of eigenfunction correlations and 
by the  decay of Fermi projections, and that the former implies finite multiplicity of the eigenvalues of the ergodic random  operator.
\end{remark}

\section{Requirements of the multiscale analysis}\label{secreq}

We now state the properties of the ergodic random operator $H_\omega$
 that are required for  the multiscale analysis and its consequence.
We will work on the continuum, but everything will work on the lattice
(easier case) with
appropriate modifications.
We fix an open interval $\mathcal{I}$.

\subsection{Generalized eigenfunction expansion}
\label{subsectgee}

Generalized eigenfunction expansions were originally developed for
elliptic partial differential operators with smooth coefficients (see
Berezanskii   \cite{B} and references therein).
  These expansions were 
extended to Schr\"odinger operators with singular potentials by
Simon \cite{simonbkg} (see also  references therein), and to classical wave operators with nonsmooth 
coefficients by Klein, Koines and Seifert \cite{KKS}. 

 These  expansions construct polynomially bounded generalized 
eigenfunctions for
a set of generalized eigenvalues  with full spectral measure.  
These generalized eigenfunctions were used
 by Pastur \cite{Pa} and by Martinelli and Scoppola \cite{MS} to prove that 
certain Schr\"odinger operators with
random potentials have no absolutely continuous spectrum.  They played 
a crucial role in the work by Fr\"ohlich, Martinelli, Spencer and Scoppola
\cite{FMSS} and by von Dreifus and Klein \cite{VDK} on exponential localization
 of random Schr\"odinger operators, providing the
crucial link between the multiscale analysis and  
pure point spectrum: the 
exponential decay of finite volume Green's functions 
(obtained by a multiscale analysis) forces polynomially bounded 
generalized eigenfunctions to be bona fide eigenfunctions, so
 the spectrum is at most countable
 and hence pure point.

In \cite{GK1},    
as in \cite{Ge,GJ}, the generalized eigenfunction
expansion itself
(not just the existence of polynomially bounded generalized 
eigenfunctions) is used
to provide the link between the multiscale analysis and 
strong HS-dynamical localization (and hence pure point spectrum). 

We will now state the properties of an ergodic random operator that  guarantees
the existence of a generalized eigenfunction expansion. 
 We follow the
approach in \cite[Section~3]{KKS}.

Let $\mathcal{H} = 
\mathrm{L}^2(\mathbb{R}^d,{\rm d}x; \mathbb{C}^n) $.  (We
 discuss the generalized eigenfunction expansion on the continuum, but an analogous discussion 
is valid on the lattice.)
Given  $\nu>d/4$ (omitted from the notation), we 
define the weighted spaces 
$\mathcal{H}_\pm$:
\begin{equation} \label{H+-}
\mathcal{H}_\pm = 
\mathrm{L}^2(\mathbb{R}^d,
\langle x\rangle^{\pm4\nu}{\rm d}x; \mathbb{C}^n)  \ .
\end{equation}
 $\mathcal{H}_-$ is a  space of polynomially 
$\mathrm{L}^2$-bounded functions.  (Recall  $\langle x \rangle= \sqrt{1+|x|^2}$.)
The sesquilinear form
\begin{equation}\label{eq:dual}
\langle\phi_1,\phi_2\rangle_{\mathcal{H}_+,\mathcal{H}_-} =\int \overline{\phi_1(x)}\cdot\phi_2(x) {\rm d}x,
\end{equation}
where $\phi_1\in\mathcal{H}_+$ and $\phi_2\in\mathcal{H}_-$, 
makes $\mathcal{H}_+$ and $\mathcal{H}_-$
 conjugate duals to each other.  By  $O^{\dagger}$ 
we will denote the adjoint
of an operator  $O$  with respect to this duality.
 By construction,
$\mathcal{H}_+\subset\mathcal{H}\subset\mathcal{H}_-\,$, the
natural injections  
 $\imath_+: \mathcal{H}_+\rightarrow\mathcal{H}$  and 
$\imath_-:\mathcal{H}\rightarrow\mathcal{H}_-$  being
continuous with dense range, with 
$\imath_+^\dagger =\imath_-\,$. 

We set  $T$ to be the self-adjoint operator  on 
$\mathcal{H}$  given by
multiplication by the function
$\langle x\rangle^{2\nu}$;  note that $T^{-1}$ is bounded.
 The operators $T_+: \mathcal{H}_+ \to \mathcal{H}$
and $T_-: \mathcal{H} \to \mathcal{H}_-$, defined by
$T_+= T \imath_+\,$, 
 $T_-$ the closure of the operator $ \imath_- T $ on $ \mathcal{D}(T)$, are unitary
with $T_-=T_+^\dagger$.
The map $\tau: \mathcal{B}(\mathcal{H})  \to 
\mathcal{ B}(\mathcal{H}_+,\mathcal{H}_-)$, with 
$\tau(C)=T_-CT_+ \,$, is a Banach space isomorphism, as $T_\pm$
are unitary operators.  ($\mathcal{ B}(\mathcal{H}_1,\mathcal{H}_2)$
denotes the Banach space of bounded operators from $\mathcal{H}_1$
to $\mathcal{H}_2$,  
$\mathcal{B}(\mathcal{H})= \mathcal{B}(\mathcal{H},\mathcal{H})$.)
 If $1\le q<\infty$, we define
$
\mathcal{ T}_{q}(\mathcal{H}_+,\mathcal{H}_-)=
\tau\left(\mathcal{ T}_{q}(\mathcal{H})\right) $, where
$\mathcal{ T}_{q}(\mathcal{H})$ denotes 
 the Banach
space of bounded operators $S$ on $\mathcal{H}$ with 
$\|S\|_q= (\mathrm{tr} \,|S|^q)^{\frac 1 q}<\infty$.
By construction,  $\mathcal{ T}_{q}(\mathcal{H}_+,\mathcal{H}_-)$, equipped with 
the norm  $\| B \|_q=\|  \tau^{-1}(B)\|_q $, is a Banach space isomorphic to 
$\mathcal{ T}_{q}(\mathcal{H})$, with
   $\mathcal{ T}_{2}(\mathcal{H}_+,\mathcal{H}_-)$ being the usual
Hilbert space of Hilbert-Schmidt operators from  $\mathcal{H}_+$
to $\mathcal{H}_-$.  

Note that
\begin{equation}\label{eq:normchi}
\|\chi_{x}\|_{\mathcal{H},\mathcal{H}_+} =
 \|\chi_{x}\|_{\mathcal{H}_-,\mathcal{H}} \leq 
\textstyle{\left(\frac 3 2\right)^\nu}\langle x\rangle^{2\nu}
\end{equation}
for all $x \in \mathbb{R}^d$.
(Given an operator $B:\mathcal{H}_1\to\mathcal{H}_2$, 
$\|B\|_{\mathcal{H}_1,\mathcal{H}_2}$ will
denote its operator norm.) 

 The following property guarantees the existence of a
 generalized eigenfunction expansion (GEE)  in the
open interval $\mathcal{I}$ 
 with the right properties (see
\cite[Section~3]{KKS} for details).  
We write $E_\omega(B)$
for the spectral projections of the operator $H_\omega$, i.e.,
$E_\omega(J)= \chi_J(H_\omega)$ for any bounded Borel set
$J\subset \mathbb{R}$.
We will fix an appropriate $\nu>d/4$ and use the corresponding operator
 $T$ and weighted
spaces $\mathcal{H}_\pm$ as in (\ref{H+-}).

\bigskip
\noindent{\textbf{(GEE)}} \ \emph{For some $\nu>d/4$ the set
\begin{equation}
\mathcal{D}_+^\omega=
\{\phi\in\mathcal{D}(H_\omega)\cap\mathcal{H}_+;
 \; H_\omega\phi\in\mathcal{H}_+\}
\end{equation}
is dense in $\mathcal{H}_+$ and   an operator core for $H_\omega$ with probability one.
Moreover, there exists  a 
 bounded, continuous  function $f$ on $\mathbb{R}$, 
strictly positive on the spectrum of $H_\omega$, such that 
\begin{equation} \label{KK11}
\textrm{tr}_{\mathcal{H}}
\left(T^{-1} f(H_\omega)E_\omega(\mathcal{I}) T^{-1}\right) <\infty \; \; 
\end{equation}
with probability one.
}
\bigskip

A measurable function $\psi: \mathbb{R}^d \to \mathbb{C}^n$ 
 is said to be a {generalized eigenfunction} of  $H_\omega$ with 
generalized eigenvalue $\lambda$, if 
$\psi \in {\mathcal{H}}_-\backslash\{0\}$  and
$$ 
\langle H_\omega\phi,\psi\rangle_{\mathcal{H}_+,\mathcal{H}_-} = 
\lambda\langle \phi,\psi\rangle_{\mathcal{H}_+,\mathcal{H}_-} \; \; 
\mbox{for all $\phi \in \mathcal{D}_+^\omega$.}
$$
It follows from the first part of property (GEE) that if a generalized eigenfunction is in
$\mathcal{H}$, then it is a bona fide eigenfunction.

If (GEE) holds,  the following is true for
 $\mathbb{P}$-almost every $\omega$: 
For all bounded Borel sets $J$ we have
\begin{equation}
\mathrm{tr}_{\mathcal{H}} 
\left(T^{-1}E_{\omega}(J\cap \mathcal{I}) T^{-1}
 \right) <+\infty \, ,
\end{equation}
and hence
\begin{equation}\label{012}
\mu_\omega (J)=
\textrm{tr}_{\mathcal{H}}
 \left(T^{-1}E_{\omega}(J\cap \mathcal{I}) T^{-1}\right)
\end{equation}
 is a spectral measure for the restriction of $H_\omega$
to the Hilbert space $E_{\omega}( \mathcal{I}) \mathcal{H}$,
with
\begin{equation} \label{eq:ingr1bis2}
  \mu_\omega (J)   < \infty \ \ \mbox{for}  \ \ J  \  \
\mbox{bounded.}
\end{equation}
In particular,
we have a generalized eigenfunction expansion for $H_\omega$: with
probability one,  
 there exists a $\mu_\omega$-locally integrable function 
$P_\omega(\lambda)$
from the real line  into ${\mathcal{ T}_1(\mathcal{H}_+,\mathcal{H}_-)}$,
 with 
\begin{equation}  \label{pdagger}
P_\omega(\lambda) =
 P_\omega(\lambda)^\dagger
\end{equation}
 and
\begin{equation}\label{a16a}
\textrm{tr}_{\mathcal{H}} \left(T_-^{-1} P_\omega(\lambda)T_+^{-1}\right)=1
 \  \ \ \mathrm{for} \ 
\mu_\omega-\mathrm{a.e.} \
\lambda\ ,
\end{equation}
such that
\begin{equation}\label{eq:expantrace}
\imath_-E_{\omega}(J\cap \mathcal{I})\imath_+=
\int_{J}P_\omega(\lambda)\,d\mu_\omega(\lambda)\  \ \
\mbox{for bounded Borel sets}\ \ J \ ,
\end{equation} 
where the integral is the Bochner integral of 
${\mathcal{ T}_1(\mathcal{H}_+,\mathcal{H}_-)}$-valued 
functions.  
Moreover, for 
$\mu_\omega$-almost every $\lambda$,
 if $\phi\in\mathcal{H}_+$ and 
$P_\omega(\lambda)\phi \not= 0$,
 then $P_\omega(\lambda)\phi$ is a generalized
eigenfunction of
$H_\omega$ with generalized eigenvalue $\lambda$. It follows,
using \eqref{eq:expantrace}, that
$\mu_\omega$-almost every $\lambda$ is a generalized eigenvalue
of $H_\omega$.

\begin{lemma}
\label{lempoly} If the ergodic random operator $H_\omega$ has property (GEE),
then  for $\mathbb{P}$-almost every $\omega$, we have
\begin{equation}\label{eq:poly2}
\|\chi_{x}P_\omega(\lambda)\chi_{y}\|_1 \leq 
\textstyle{\left(\frac 3 2\right)^{2\nu}}
\langle x \rangle^{2\nu} \langle y \rangle^{2\nu}
\end{equation}
 for all $x,y\in\mathbb{R}^d$ and
 $\mu_\omega$-almost every $\lambda$. ($\| \ \|_1$ denotes
the trace norm in $\mathcal{H}$.)
\end{lemma}

\begin{proof}
Since
\begin{eqnarray} \label{eq:poly1}
\|\chi_{x}P_\omega(\lambda) \chi_{y}\|_1\leq 
 \|\chi_{x}\|_{\mathcal{H}_-,\mathcal{H}}
 \|P_\omega(\lambda)\|_{\mathcal{ T}_1(\mathcal{H}_+,\mathcal{H}_-)}
\|\chi_{y}\|_{\mathcal{H},\mathcal{H}_+} \, ,
\end{eqnarray}
(\ref{eq:poly2}) follows from (\ref{eq:normchi}) and (\ref{a16a}). 
\end{proof}

(GEE) suffices for proofs of exponential localization \cite{FMSS,VDK} 
and dynamical localization \cite{GDB,Ge}.  But for strong
 dynamical localization we need to strengthen (\ref{KK11}).

\bigskip
\noindent{\textbf{(SGEE)}} \ Property \emph{(GEE) holds
with
\begin{equation} \label{KK1}
\mathbb{E}\left\{\left[\textrm{tr}_{\mathcal{H}}
\left(T^{-1} f(H_\omega)E_{\omega}( \mathcal{I}) T^{-1}\right)\right]^2\right\} <\infty  \ .
\end{equation}
}

\medskip
It follows that 
\begin{equation}\label{eq:ingr1}
\mathbb{E}\left\{\left[\mathrm{tr}_{\mathcal{H}} \left(T^{-1}
E_{\omega}(J\cap \mathcal{I}) T^{-1}
 \right)\right]^2 \right\} < \infty
\end{equation}
for all bounded Borel sets $J$, so 
we have a stronger version of (\ref{eq:ingr1bis2}):
\begin{equation} \label{eq:ingr1bis}
 \mathbb{E}\left\{\left[ \mu_\omega (J)  \right]^2\right\} < \infty \ \ \mbox{for}  \ \ J  \  \
\mbox{bounded.}
\end{equation}

\begin{remark} Estimate (\ref{KK1}) is true for the
 usual ergodic random operators.  In fact one usually proves
 the stronger 
\begin{equation} \label{KK12}
\left\| \textrm{tr}_{\mathcal{H}} 
\left(T^{-1} f(H_\omega) E_{\omega}( \mathcal{I})T^{-1}\right)
 \right\|_{\mathrm{L}^\infty(\Omega, \mathcal{F},\mathbb{P})} < \infty \, ,
\end{equation}
 which is a hypothesis in \cite{DS}.  For a proof,
see   \cite[Theorem~1.1]{KKS} for classical wave operators
and   \cite{simonbkg},\cite[Theorem A.1]{GK3} for Schr\"odinger operators.
\end{remark}

\subsection{Finite volume operators and their properties} \label{notations}

Throughout these lectures we use the sup norm in $\mathbb{R}^d$: 
\begin{equation}|x|= \max \{|x_i|,\;i=1,\ldots,d\} \ .
\end{equation}
By $\Lambda_L(x)$  we denote the
 open box (or cube) of side $L>0$ centered at
 $x\in\mathbb{R}^d$:
\begin{equation}
\Lambda_L(x) =\left \{y\in\mathbb{R}^d  ; \ \, |y-x|<\textstyle{\frac L 2}\right\} ,
\end{equation}
and by $\overline{\Lambda}_L(x)$ the closed box.
We set
\begin{equation}
\chi_{x,L} = \chi_{\Lambda_L(x)}, \quad \chi_x = \chi_{x,1}
=\chi_{\Lambda_1(x)}.
\end{equation}

 \emph{We will usually take boxes
centered at sites $x \in\mathbb{Z}^d$ with side  $L\in2\mathbb{N}$.} 
Given such a box  $\Lambda_L(x)$, we set
\begin{equation}
\Upsilon_{L}(x) = \left\{y \in \mathbb{Z}^d; \ |y-x| = \textstyle{\frac L 2} - 1  \right\},
\end{equation}
 and define its boundary belt  by
\begin{equation} \label{bdrybelt}
\tilde{\Upsilon}_{L}(x) =  \overline{\Lambda}_{L-1}(x) \backslash 
{\Lambda}_{L-3}(x)  = 
\bigcup_{y \in \Upsilon_{L}(x) } \overline{\Lambda}_1(y) \ ;
\end{equation}
it has the characteristic function
\begin{equation}
\Gamma_{x,L} = \chi_{\tilde{\Upsilon}_{L}(x)} =
 \sum_{y \in \Upsilon_{L}(x) } \chi_y \; \; a.e.
\end{equation}
Note that
\begin{equation}
|\Upsilon_L(x)| = (L-1)^d -(L-2)^d = d\int_{L-2}^{L-1} x^{d-1} dx 
\le d (L-1)^{d-1}\, .
\end{equation}

We shall suppress the dependency of a box on its center  when not necessary.
When using boxes $\Lambda_{\ell}$  contained in bigger boxes
$\Lambda_L$, we shall need to know that the small box is inside the
belt $\tilde{\Upsilon}_L$ of the bigger one. If
  $L > \ell + 3 $ and $x\in \mathbb{Z}^d$, 
we  say that
\begin{equation}
\Lambda_{\ell}\sqsubset\Lambda_{L}(x) \quad \mbox{if}
\quad\Lambda_{\ell}\subset
\Lambda_{L-3}(x)\ .
\end{equation}

Very often we will require $L\in 6\mathbb{N}$;   given $K \ge 6$, we set
\begin{equation}
[K]_{6\mathbb{N}}= \max \{ L \in 6\mathbb{N}; \; L \le K\}.
\end{equation}

The multiscale analysis requires the notion of a 
{\em finite volume operator}, a ``restriction"
$H_{\omega,x,L}$ of $H_{\omega}$ to the box
$\Lambda_L(x)$  where the ``randomness based
 outside the box $\Lambda_L(x)$" is not taken into account. 
Usually
 $H_{\omega,x,L}$ is defined as the restriction of $H_\omega$, either
to the open  box $\Lambda_L(x)$ with 
 Dirichlet boundary condition, or to the closed box
$\overline{\Lambda}_L(x)$ with  periodic boundary condition. The
operator $H_{\omega,x,L}$ then acts on
$\mathrm{L}^2(\Lambda_L(x),{\rm d}x; \mathbb{C}^n)$.  But   
$H_{\omega,x,L}$ may also be defined as acting on the whole space,
by throwing away the random coefficients ``based
 outside the box $\Lambda_L(x)$"; this is usually done for random
Landau operators \cite{CH2,W1,GK4}. In all cases the
finite volume operators have either compact resolvent or 
are  relatively compact perturbations of the free Hamiltonian.

\begin{definition} The ergodic random operator $H_\omega$ is called standard
if it has a finite volume
restriction, i.e.,  if for each  $x \in\mathbb{Z}^d$ and  $L\in2\mathbb{N}$
there is a measurable map
 $H_{\omega,x,L}$
from the probability space $(\Omega,\mathcal{F},\mathbb{P})$
 to self-adjoint operators on
$\mathrm{L}^2(\Lambda_L(x),{\rm d}x; \mathbb{C}^n)$
(or all such mappings taking values as  self-adjoint operators on 
 $\mathrm{L}^2(\mathbb{R}^d,{\rm d}x; \mathbb{C}^n)$), such that 
\begin{equation}
U(y) H_{\omega,x,L} U(y)^*=  H_{\tau_y(\omega), x+y,L}\;\; 
\mbox{for  all $ y \in \mathbb{Z}^d$},
\end{equation}
where $U(y)$ is as in \eqref{cov}.
  We write
$R_{\omega,x,L}(z)= (H_{\omega,x,L}-z)^{-1}$ for the resol vent of
the finite volume operator $H_{\omega,x,L}$
 and $E_{{\omega,x,L}}(\cdot)$ for its spectral projection.
\end{definition}

The multiscale analysis and its consequences require
 certain properties
of the finite volume restriction of the ergodic random operator.
 These properties
are routinely verified for the usual ergodic random operators
 (e.g.,  \cite{FS,FMSS,VDK,HM,CH1,CH2,FK1,FK2,W1,St,KK1,KK2,GK3,GK4}. 

The first property is  \emph{independence at a distance} (IAD)
for the finite volume operators.
 It says that 
if boxes are far apart, events defined by the restrictions of the random
operator $H_\omega$ to these  boxes are independent.  
 This assumption can be relaxed in some ways by suitable
 modifications of the multiscale analysis (e.g., 
\cite{VDK2,KSS2,FLM,Z}).

An 
event is said to be based on
the box $\Lambda_L(x) $ if it is determined by conditions on the 
finite volume operator
$H_{\omega,x,L}$.  Given  $\varrho > 0$, we say that two boxes
$\Lambda_L(x) $ and $\Lambda_{L^\prime}(x^\prime) $  are
\emph{$\varrho$-nonoverlapping} if  
$|x - x^\prime| > \frac{L + L^\prime} 2  + \varrho$,  i.e., if
$\mathrm{dist}(\Lambda_L(x),\Lambda_{L^\prime}(x^\prime)  ) 
>\varrho$. 
\bigskip

\noindent{\textbf{(IAD)}} \  
\emph{There exists $\varrho > 0$ such that events based on 
$\varrho$-non\-overlapping boxes are independent.}

\smallskip
The remaining properties are to hold in the fixed  open interval $\mathcal{I}$.

The first such property is reminiscent of the Simon-Lieb inequality (SLI) in
Classical Statistical Mechanics.  It relates resolvents in different scales.
In the lattice it is an immediate
consequence of the resolvent identity, in this context it was originally
used in \cite{FS}.  In the continuum, its proof requires interior estimates, and
was   proved in \cite{CH1}
for Schr\"odinger operators.
 It was adapted to
 classical wave operators in \cite{FK1}.  We state it in the form given in
\cite[Lemma 3.8]{KK1} for  classical wave operators and 
\cite[Theorem A.1]{GK3} for Schr\"odinger operators. 
 (The lattice requires slight modifications.)

\bigskip\noindent{\textbf{(SLI)}} \ 
\emph{For any compact interval $I \subset \mathcal{I}$  
there exists a finite constant $\gamma_{I}$,
such that,  given $ L,\ell^\prime,\ell^{\prime\prime}\in 2\mathbb{N}$,
  $x, y, y^\prime\in \mathbb{Z}^d$ with
$\Lambda_{\ell^{\prime\prime}}(y) 
\sqsubset \Lambda_{\ell^\prime}(y^\prime) \sqsubset
 \Lambda_{L}(x)$, then
for $\mathbb{P}$-almost every $\omega$, if $E \in {I}$ with
 $E \notin
\sigma(H_{\omega,x,L})\cup  \sigma(H_{\omega,y^\prime,\ell'})$,
   we have 
\begin{equation} 
\| \Gamma_{x,L} R_{\omega,x,L}(E) \chi_{y,\ell^{\prime\prime}}\|
\le \gamma_{{I}}\,
\| \Gamma_{y^\prime,\ell'} R_{\omega,y',\ell'}(E) 
\chi_{y,\ell^{\prime\prime}}\| \,
\| \Gamma_{x,L} R_{\omega,x,L}(E)
 \Gamma_{y^\prime,\ell'}\| \,   . \label{sli1}
\end{equation}}

\begin{remark} Property (SLI)  will be used  in the following way: 
 We will take 
$ \ell^{\prime\prime} =\frac \ell 3$ with $\ell \in 6 \mathbb{N}$, and 
$\ell'=k\frac \ell 3$ with $3\le k \in \mathbb{N}$.  By a  \emph{cell}
 we will mean a closed box 
$\overline{\Lambda}_{\frac \ell 3}(y^{\prime\prime})$, with 
$y^{\prime\prime} \in \frac \ell 6\mathbb{Z}^d$.  We define
 $\mathbb{Z}_{\mathrm{even}}$
and $\mathbb{Z}_{\mathrm{odd}}$ to be the sets of even and odd
integers.  We take
 $y \in \frac \ell 6\mathbb{Z}^d$, so $\chi_{y,\frac \ell 3}$ is the characteristic
function of a cell. We want 
 the closed box 
$\overline{\Lambda}_{\ell'}(y')$ to be exactly
covered by   cells (in effect, by $k^d$ cells); thus
 we specify
$y^\prime  \in \frac \ell 3\mathbb{Z}^d= 
\frac \ell 6\mathbb{Z}_{\mathrm{even}}^d$  if $k$ is odd, and 
$y^\prime \in \frac\ell 3\mathbb{Z}^d  + 
 \frac \ell 6\left(1,1, \ldots, 1\right)=
\frac \ell 6\mathbb{Z}_{\mathrm{odd}}^d$
 if $k$ is even.
We then replace the boundary belt $ \tilde{\Upsilon}_{\ell'}(y')$ 
(of width $1$) by a thicker
belt $\tilde{\Upsilon}_{\ell',\ell}(y')$ of width $\frac \ell 3$.  To do so,
we set
\begin{equation} \label{2Upsilon}
\Upsilon_{\ell',\ell}(y') = 
\left\{y^{\prime\prime} \in \frac \ell 3\mathbb{Z}^d; \ |y^{\prime\prime}-y'| 
= \frac {\ell'}2 - \frac {\ell}6  \right\},
\end{equation}
and define the boundary $\ell$-belt of $ \Lambda_{\ell'}(y') $ 
 by
\begin{equation}
\tilde{\Upsilon}_{\ell',\ell}(y')=  \overline{\Lambda}_{\ell'}(y') \backslash 
{\Lambda}_{\ell' - \frac{2\ell}3}(y')  = 
\bigcup_{y^{\prime\prime} \in \Upsilon_{\ell',\ell}(y') }
\overline{\Lambda}_{{\frac \ell 3}}(y^{\prime\prime}) \ ,
\end{equation}
with characteristic function
\begin{equation}
\Gamma_{y',\ell',\ell} = \chi_{\tilde{\Upsilon}_{\ell',\ell}(y')} =
 \sum_{y^{\prime\prime} \in\Upsilon_{\ell',\ell}(y')  } \chi_{y^{\prime\prime}, {{\frac \ell 3}}} \; \; a.e.
\end{equation}
Note that
\begin{equation}
|\Upsilon_{\ell',\ell}(y')|= (k^d-(k-2)^d) \le k^d  \, .
\end{equation}
Since
$\Gamma_{y',\ell',\ell}\Gamma_{y',\ell'}= \Gamma_{y',\ell'}$,
the projection $\Gamma_{\ell'}$ on the belt of $\Lambda_{\ell'}$ can be
replaced by the projection over the thicker belt of width ${\frac \ell 3}$, which can
be decomposed in boxes of side ${\frac \ell 3}$.  Thus (\ref{sli1}) yields
\begin{eqnarray} \label{sli}
\| \Gamma_{x,L} R_{\omega,x,L}(E) \chi_{y,{\frac \ell 3}}\| 
\!\le\!
k^d\gamma_{I}   
\| \Gamma_{y',\ell'} R_{\omega,y',\ell'}(E) \chi_{y,{\frac \ell 3}}\| 
\| \Gamma_{x,L} R_{\omega,x,L}(E) \chi_{y^{\prime\prime},{{\frac \ell 3}}}\|
\ \ 
\end{eqnarray} 
for some $y^{\prime\prime} \in \Upsilon_{\ell',\ell}(y')$. 
Performing the SLI, i.e., using the estimate (\ref{sli}), we moved from
the cell center $y$ to the cell center $y^{\prime\prime}$.
\end{remark}

\begin{remark} \label{rsli}
While performing a multiscale analysis we will use (\ref{sli})
  with either $\ell'=\ell$ (for good boxes), or
some  $\ell'=k{\frac \ell 3}$, $k >3$, which will be the side of a bad box. Note that in the first case,
$k=3$, and the geometric factor is $3^d-1 \leq 3^d$. In that case note also
that we must have $y=y^\prime$ and
$|y^{\prime\prime}-y|={\frac \ell 3}$, so after performing the SLI
we moved to an adjacent cell center, i.e., by ${\frac \ell 3}$ in the sup norm.  (Recall that 
we are using the sup norm in $\mathbb{R}^d$, so we may
 move both sidewise and along the diagonals.)
\end{remark}

The second property is an estimate of generalized eigenfunctions 
in terms of finite volume resolvents.  It is not needed for the
 multiscale analysis,
but it plays an important role in obtaining localization from
the multiscale analysis \cite{FMSS,VDK,FK1,GK1}. We call it an \emph{eigenfunction
decay inequality} (EDI), since it translates decay of finite volume resolvents
 into decay of generalized eigenfunctions ; we present it as proved in 
\cite[Lemma 3.9]{KK1} and \cite[Theorem A.1]{GK3}.
  It is closely related to property (SLI), the proofs
being very similar.
\bigskip

\noindent{\textbf{(EDI)}} \ 
\emph{For any compact interval $I \subset \mathcal{I}$  
there exists a finite constant $\tilde{\gamma}_{I}$,
such that for $\mathbb{P}$-almost every $\omega$,   given   a generalized 
eigenfunction $\psi$ of $H_\omega$  with generalized eigenvalue 
 $E \in  I$,
we have for any
$x \in  \mathbb{Z}^d$ and
$L\in 2\mathbb{N}$ with
 $E \notin \sigma(H_{\omega,x,L})$ that
\begin{equation}
\| \chi_x \psi \| \le  \tilde{\gamma}_{I} 
\| \Gamma_{x,L} R_{\omega,x,L}(E) \chi_{x}\|
\|  \Gamma_{x,L} \psi\| \ .\label{ediNew}
\end{equation}
}

\medskip
Typically we have $\tilde{\gamma}_{I} =\gamma_{I} $, with 
$\gamma_{I}$ as in (\ref{sli1}). We will use the following consequence of
 (\ref{ediNew}):
\begin{equation}
\| \chi_x \psi \| \le  d \tilde{\gamma}_{I} L^{d-1}
\| \Gamma_{x,L} R_{\omega,x,L}(E) \chi_{x}\|
\|  \chi_{y} \psi\| \label{edi}
\end{equation}
for some $y \in \Upsilon_{L}(x)$.

The third property is an ``a priori" estimate
on the  average \emph{number of eigenvalues} (NE) of 
finite volume random operators in a fixed, bounded interval.  It is
 usually proved by a deterministic
argument, using the well known bound for the  Laplacian 
\cite{CH1,FK1,FK2,KK1}.
It is, of course, entirely obvious in the lattice.
\bigskip

\noindent{\textbf{(NE)}} \  
\emph{For any compact interval $I \subset \mathcal{I}$  
there exists a finite constant $C_{I}$ such that  
\begin{equation} \label{density}
\mathbb{E} \left(\mathrm{tr}_\mathcal{H} 
E_{{\omega,x,L}}({I}) \right) \le C_{I} L^d 
\end{equation}
for all $x \in \mathbb{Z}^d$ and 
$L \in 2\mathbb{N}$.}

\bigskip
 The final property is a form of 
\emph{Wegner's estimate} (W),
 a probabilistic  estimate on the size of the resolvent. 
 It is a crucial ingredient for the multiscale analysis, where it is used
 to control the bad regions.
\bigskip

\noindent{\textbf{(W)}} \  
\emph{ For some $b \ge1$ there exists  
 a finite constant
  $Q_{I} $ for each compact interval $I \subset \mathcal{I}$,  
 such that 
\begin{equation}
\mathbb{P}\left\{ \mbox{\rm dist}(\sigma (H_{\omega,x,L}),E)\le \eta
\right\} \le Q_{I} \eta L^{bd} \ , \label{wegner}
\end{equation}
 for all  $E \in{I}$,  $0<\eta\le 1$, $x \in \mathbb{Z}^d$, and 
$L \in 2\mathbb{N}$.}

\begin{remark} \label{apresW}
In practice we have either $b=1$ or $b=2$ in the Wegner
estimate (\ref{wegner}). For some random Schr\"odinger operators
with Anderson potential
we may have $b=1$
\cite{CH1,Kl} (including the Landau Hamiltonian). For classical waves in
random media, (\ref{wegner}) has been proven with $b=2$
\cite{FK1,FK2,KK2}.
More  recently the correct
volume dependency (i.e., $b=1$) 
was obtained in \cite{CHN,CHKN,HK} for random Schr\"odinger operators, at the price
of  losing a  bit in
the $\eta$ dependency; more precisely, the right hand side of
(\ref{wegner}) is replaced by $ Q_{a,I} \eta^a L^{d} $ for
any $0<a <1$.  In these lectures,
we shall use (\ref{wegner}) as stated, the  modifications in our methods
required for the other
forms of (\ref{wegner}) being obvious.
Our methods may also accomodate properties (NE) and (W) being valid only
for large $L$, and/or property (W) being valid only for
 $\eta <\eta_L$ for some appropriate $\eta_L$, say $\eta_L = L^{-r}$,
some $r>0$,
or $\eta_L= \mathrm{e}^{-L^\beta}$ for some $0<\beta<1$.
 The latter is of importance if one wants to deal with singular probability
measures like Bernoulli \cite{CKM,KLS,DBG,DSS1}.
\end{remark}

\begin{remark} \label{apresW3}
In the continuum one usually proves the stronger estimate
\cite{HM,CH1,CH2,FK1,FK2,KK2,CHN}:
\begin{equation}
\mathbb{E} \left(\mathrm{tr}_\mathcal{H}
E_{H_{\omega,x,L}}\left([E-\eta,E+\eta]\right) \right)
\le Q_{I} \eta L^{bd} \ , \label{wegner+}
\end{equation}
from which (\ref{wegner}) follows by Chebychev's inequality.
The estimate (\ref{density}) is used as an ``a priori" estimate
in the proof of (\ref{wegner+}).
\end{remark}

\section{The  bootstrap multiscale analysis} \label{secMSA}

  Given a standard ergodic random operator $H_\omega$,
 the multiscale analysis  looks for localization  by studying the 
probability of decay
 of the  finite volume resolvent from the center of a box  $\Lambda_L(x)$ to its
boundary belt as measured by 
\begin{equation}
\| \Gamma_{x,L} R_{\omega,x,L}(E) \chi_{x,{{\frac L 3}}}\| \, .
\end{equation}

We start with three definitions, which characterize ``good boxes" in a given scale
by different types of decay relative to the scale.

\begin{definition} \label{dsuitable}
Given $E \in \mathbb{R}$,
$x \in  \mathbb{Z}^d$ and $L \in 6 \mathbb{N}$,  with
$E \notin \sigma(H_{\omega,x,L})$, 
we say that
 the box 
$\Lambda_L(x) $ is
\begin{description}
\item[(i)\!\!] $(\omega,\theta, E)$-suitable for a given $\theta >0$  if 
\begin{equation}
\| \Gamma_{x,L} R_{\omega,x,L}(E) \chi_{x,{{\frac L 3}}}\| 
 \le \frac{1}{L^\theta} \,  .
\label{suitable}
\end{equation}

\item[(ii)\!\!]  $(\omega,\zeta, E)$-sub-exponentially-suitable for a
 given $\zeta\in(0,1)$ if
\begin{equation}
\| \Gamma_{x,L} R_{\omega,x,L}(E) \chi_{x,{{\frac L 3}}}\| 
\le\mathrm{e}^{-L^\zeta} \,  .
\label{subexp_suitable}
\end{equation}

\item[(iii)\!\!]  $(\omega,m, E)$-regular  for a
 given $m>0$ if
\begin{equation} 
\| \Gamma_{x,L} R_{\omega,x,L}(E)\chi_{x,{{\frac L 3}}} 
\| \le {\rm e}^{-m\frac{L}{2}} 
\, . \label{regular}
\end{equation}

\end{description}
\end{definition}

\begin{remark} \label{regsuit}
Note that a box $\Lambda_L(x) $ is $(\omega,\theta, E)$-suitable if and only if
it is $(\omega, m, E)$-regular, where $m=2\theta\frac{\log L} L$.   Similarly,
$\Lambda_L(x) $ is $(\omega,\zeta, E)$-sub-exponentially-suitable
 if and only if
it is $(\omega, 2L^{\zeta-1}, E)$-regular.

\end{remark}

The multiscale analysis converts decay with high  probability at a large enough scale
into  decay with better probabilities at higher scales. We state 
 the strongest version, the bootstrap multiscale analysis
of Germinet and Klein \cite[Theorem 3.4]{GK1}.

\begin{definition} \label{sigmaSMA}  Let $H_\omega$ be a 
 standard ergodic random operator with property (IAD).
The multiscale analysis  region
$\Sigma_{\mathrm{MSA}}$ for 
 $H_\omega$ is the set of $E \in\Sigma$ for which there exists some
open interval $I \ni E$, such that 
 given any $\zeta$, $0<\zeta<1 $, and $\alpha$,
 $ 1<\alpha<\zeta^{-1}$, there is a length scale
$L_0\in 6 \mathbb{N}$
 and a mass $m>0$, so if we
 set 
$L_{k+1} = [L_k^\alpha]_{6\mathbb{N}}$, $k=0,1,\dots$,
 we have
\begin{equation} \label{msres-bstp}
\mathbb{P}\,\left\{R\left(m, L_k, I,x,y\right) \right\}\ge 1
-\mathrm{e}^{-L_{k}^\zeta}
\end{equation} for all $k=0,1,\ldots$, and  $x, y \in \mathbb{Z}^d$
with 
$|x-y| > L_k +\varrho$, where 
\begin{eqnarray}
\lefteqn{R(m,L, I,x,y) =} \label{eq:defsetR} \\ 
&&\{\mbox{$\omega$; for every}
\; E^\prime \in I
\ \mbox{either} \ \Lambda_L(x) \ \mbox{or} \ 
\Lambda_L(y) \  \mbox{is} \  \mbox{$(\omega,m, E^\prime)$-regular} \} \ .
\nonumber
\end{eqnarray} 
\end{definition}

\begin{theorem}[{\cite[Theorem 3.4]{GK1}}] \label{tBMSA}  Let $H_\omega$ be a 
 standard ergodic random operator with  (IAD) and properties
 (SLI), (NE) and (W) in an open interval $\mathcal{I}$.  Given
 $\theta > bd$, for each $E \in \mathcal{I}$   there exists 
 a  finite scale 
${\mathcal{L}}_\theta (E)={\mathcal{L}}_\theta (E,b,d,\varrho)$,
 bounded on compact subintervals of 
$\mathcal{I}$, 
 such that, if  for a given
  $E_0 \in \Sigma\cap {\mathcal{I}} $   we
can verify  that
\begin{equation}
\mathbb{P}\{\Lambda_{L_0}(0) \;\;
\mbox{is $(\omega,\theta, E_0)$-suitable} \}
> 1 - \frac 1 {841^d} \ 
 \label{hypH}
\end{equation}
at some scale $L_0\in 6 \mathbb{N}$ with
$L_0 > {\mathcal{L}}_\theta (E_0)$, 
then $E_0 \in \Sigma_{\mathrm{MSA}}$.
\end{theorem}

\begin{remark} Explicit estimates on
 ${\mathcal{L}}_\theta (E)$ are given
in \cite{GK4}.
\end{remark}

We call Theorem~\ref{tBMSA} the bootstrap multiscale analysis because its proof
uses  four different multiscale analyses, each one bootstrapping into the next. 
We present them in the order in which they are used.

\begin{theorem}[{\cite[Lemma 36]{FK1}, \cite[Theorem 5.1]{GK1}}]  \label{thmmsa1}
Let $H_\omega$ be a standard ergodic random operator with (IAD)
and properties (SLI) and 
(W) in an  open interval $\mathcal{I}$.  Let $I_0$ be a compact subinterval of
$\mathcal{I}$, 
$E_0 \in I_0$, and  $\theta > bd$.  Given  an odd integer
$Y \ge 11$,  for any $p$ with
   $0 <p < \theta - bd $ we can find
${\mathcal{Z}} = {\mathcal{Z}}( d,\varrho,Q_{I_0},  \gamma_{I_0},b,\theta, p, Y) $,
such that if for some $L_0>{\mathcal{Z}}$, $L_0\in 6\mathbb{N}$, we have
 \begin{equation} \label{hypH2} 
\mathbb{P}\{\Lambda_{L_0}(0) \;\;\mbox{is $(\theta, E_0)$-suitable} \}
 > 1 -  (3Y-4)^{-2d},
\end{equation}
 then, setting $L_{k+1}= Y L_k$, $k=0,1,2,\ldots$,  
 we have that   
  \begin{equation}
\mathbb{P}\{\Lambda_{L_k}(0) \;\;\mbox{is $(\theta, E_0)$-suitable} \}
 \ge 1 - \frac{1}{L_k^p} 
\label{x11}
\end{equation}
for all $k \ge {\mathcal{K}}$, where $  {\mathcal{K}}= {\mathcal{K}}(p,Y,L_0) < \infty$.
\end{theorem}

The value of Theorem~\ref{thmmsa1} is that it requires a very weak
starting hypothesis, in which the bound on the probability of the bad event is
independent of the scale, and its conclusion, 
 in view of Remark~\ref{regsuit}, gives the starting hypothesis of  a modified
 form of the 
 Dreifus--Klein
multiscale analysis,  Theorem~\ref{msa2} below. 
Theorem~\ref{thmmsa1} is an enhancement of
{\cite[Lemma
36]{FK1}}, adapted to our  assumptions and definitions. 
It is  proven  by a multiscale analysis which combines  an idea of Spencer
 \cite[Theorem 1]{Sp}
with the methods of \cite{VDK}.

\begin{theorem}[{\cite[Theorem 32]{FK1},\cite[Theorem 5.2]{GK1}}] \label{msa2}
Let $H_\omega$ be a standard ergodic random operator 
with (IAD)
and properties (SLI) and 
(W) in an  open interval $\mathcal{I}$.   Let
 $I_0$ be a compact subinterval of  $\mathcal{I}$,
$E_0 \in I_0$,  $\theta > bd$, and   $0< p < \theta - bd $.
Then given $p^\prime > p$ and  
 $1<\alpha< \min\left\{\frac {2{p}+2d}{{p} + 2d}, 
\frac{\theta }{{p} +bd}\right\}$,  there is
${\mathcal{B}}=
{\mathcal{B}}(d,b,\varrho, Q_{I_0},\gamma_{I_0},\theta,p,p', \alpha) $,
 such that, if at some finite scale $L_0\geq {\mathcal{B}}$ we verify that
\begin{equation}
\mathbb{P}\{\Lambda_{L_0}(0) \;\;\mbox{is 
$(2\theta {\log L_0 \over L_0}, E_0)$-regular} \} \ge 1 - \frac{1}{L_0^{p^\prime}} \ ,
\label{p1}
\end{equation}
then there exists  
 $\delta_1=
\delta_1(d,b,\theta,p, \alpha,L_0)>0$,
such that if we set
   $I(\delta_1)=[E_0-\delta_1, E_0+\delta_1]\cap I_0$,
$m_0= 2\theta {\log L_0 \over L_0}$, and
$L_{k+1} = [L_k^\alpha]_{6\mathbb{N}}$, $k=0,1,\dots$,
we have 
\begin{equation}
\mathbb{P}\{\Lambda_{L_k}(0) \;\;\mbox{is 
$(\frac{m_0}{2}, E)$-regular} \} 
\ge 1 - \frac{1}{L_k^{p}}  \; \; \; \; \mbox{for all} \; \;   E\in I(\delta_1)\, ,
\label{p111th}
\end{equation}
for  all $k=0,1,\ldots$.

If in addition $H_\omega$ has property (NE) in  $\mathcal{I}$ and we have
 $\theta > 2p +(b+1)d $,  then, fixing   a
compact subinterval ${\tilde{I}_0}$ of $\mathcal{I}$ with
$I_0 \subset \tilde{I}_0^\circ $,   there is a scale   
$\widetilde{\mathcal{B}}=
\widetilde{\mathcal{B}}(d,b,\varrho, Q_{{\tilde{I}_0}}, C_{\tilde{I}_0},
\gamma_{I_0}, \mathrm{dist}(I_0,\mathcal{I}\backslash{\tilde{I}_0}), 
\theta,p,p', \alpha) $,
 such that, if at some finite scale $L_0\geq \widetilde{\mathcal{B}}$
we verify \eqref{p1}, we  have
\begin{equation} \label{msres-bstp2}
\mathbb{P}
\left\{R\left({\textstyle{\frac{m_0}{2}}}, L_k, I(\delta_1),x,y\right) \right\}\ge 1 - 
\frac{1}{L_k^{2p}} \;\; \mbox{for all}\;  \; x, y \in \mathbb{Z}^d\, , 
|x-y| > L_k +\varrho \ , 
\end{equation}
for all $k=0,1,\ldots$.
\end{theorem}

 Theorem~\ref{msa2} is an enhancement of the  Dreifus-Klein multiscale
analysis   \cite{VDK}. The crucial difference  is that
Theorem~\ref{msa2} allows the mass to go to zero as the inital scale $L_0$
goes to infinity, which may seem very surprising at the first sight. Indeed, in
the original versions of the MSA ( e.g., \cite{FS,FMSS,VD,VDK,CH1}), the
mass {\it has} to be fixed first in order to know how large $L_0$ has to be
chosen.   Figotin and Klein  {\cite[Theorem 32]{FK1}} were the first to note that the
 mass may depend on the 
scale, as  in
(\ref{p1}) above,  i.e., a mass proportional to $\frac{\log L_0} {L_0}$. 
\emph{Thus the starting hypothesis (\ref{p1}) only requires the decay of the resolvent
 on finite boxes  to be polynomially small  in the scale, 
  not 
exponentially small.} Note also that by using the SLI as in (\ref{sli}), so
 we only move between cells, we only need to require $p>0$ as in \cite{KSS}, not $p>d$
as in \cite{VDK}  (we need to consider only the
 $\left(3 \frac L \ell\right)^d$ cells that are cores of boxes of side $\ell$ inside the bigger
box of side $L$, instead of $ L^d$ boxes as in  \cite{VDK}).

Only  the weaker conclusion (\ref{p111th}) is needed for
 the bootstrap multiscale
analysis;  we also stated (\ref{msres-bstp2}) because
 it is the usual conclusion of this
multiscale analysis.  Note that for  conclusion (\ref{p111th}) we may take
$p^\prime =p$ with $\delta_1=0$.

Theorems~~\ref{thmmsa1}  and \ref{msa2}
only yield polynomially decaying 
probabilities for bad events.   Germinet and Klein \cite{GK1} introduced 
 new versions of these
multiscale analyses that give sub-exponential decay for the probabilities of
 bad events.

\begin{theorem}
\label{thmsuit-subexp}
Let $H_\omega$ be a standard ergodic random operator with (IAD)
and properties (SLI) and 
(W) in an  open interval $\mathcal{I}$.  Let
 $I_0$ be a compact subinterval of
$\mathcal{I}$,
$E_0 \in I_0$, and  $\zeta_0\in(0,1)$.
 Given  an odd integer
$Y \ge 11^{\frac 1{1-\zeta_0}}$,  for any $\zeta_1$ with $0<\zeta_1<\zeta_0$
we can find
${\mathcal{Z}} = {\mathcal{Z}}( d,\varrho,Q_{I_0},  \gamma_{I_0},b,\zeta_0,\zeta_1, Y) $,
such that if for some $L_0>{\mathcal{Z}}$, $L_0\in 6\mathbb{N}$, we have
 \begin{equation}
\mathbb{P}\{\Lambda_{L_0}(0) \;\;\mbox{is $(\zeta_0, E_0)$-sub-exponentially-suitable} \} 
> 1 - (3Y-4)^{-2d} ,
\label{hyp1_subexp}
\end{equation}
 then, setting $L_{k+1}= Y L_k$, $k=0,1,2,\ldots$,  
 we have that   
\begin{equation}
\mathbb{P}\{\Lambda_{L_k}(0) \;\;\mbox{is $(\zeta_0, E_0)$-sub-exponentially-suitable} \}
 \ge 1 -\mathrm{e}^{-L_k^{\zeta_1}} 
\label{xxk_subexp}
\end{equation}
  for all $k \ge {\mathcal{K}}$, where $  {\mathcal{K}}= {\mathcal{K}}(\zeta_0,\zeta_1,Y,L_0) < \infty$. 
\end{theorem}

\begin{theorem} \label{thmmsa-subexp}
Let $H_\omega$ be a standard ergodic random operator with (IAD)
and properties (SLI), (NE) and 
(W) in an  open interval $\mathcal{I}$.   Let
 $I_0$ be a compact subinterval of
$\mathcal{I}$,  
$E_0 \in I_0$,  ${\tilde{I}_0}$ a
compact subinterval  of $\mathcal{I}$ with
$I_0 \subset \tilde{I}_0^\circ $, and $0<\zeta_2<\zeta_1<\zeta_0<1$.  Then,
 given
 $1<\alpha<\zeta_0/\zeta_1$, there is 
${\mathcal{C}}=
{\mathcal{C}}(d,b,\varrho, Q_{{\tilde{I}_0}}, C_{\tilde{I}_0},
\gamma_{I_0}, \mathrm{dist}(I_0,\mathcal{I}\backslash{\tilde{I}_0}),
\zeta_0,\zeta_1, \zeta_2,\alpha)$,
such that, if at some finite scale $L_0 \ge{\mathcal{C}}$, $L_0 \in 6\mathbb{N}$, we verify that
\begin{equation} 
\mathbb{P}\{\Lambda_{L_0}(0) \;\;\mbox{is 
$(2L_0^{\zeta_0-1}, E_0)$-regular} \} \ge 1 -\mathrm{e}^{-L_0^{\zeta_1}} \, ,
\label{p1_subexp}
\end{equation}
then there exists 
 $\delta_2=
\delta_2(\zeta_0,\zeta_1, L_0)>0$ such that, if we set
  $I(\delta_2)=[E_0-\delta_2, E_0+\delta_2]\cap I_0$,
$m_0= 2L_0^{\zeta_0-1}$, and
$L_{k+1} = [L_k^\alpha]_{6\mathbb{N}}$, $k=0,1,\dots$,
we have 
\begin{equation} \label{msres_subexp}
\mathbb{P}\left\{R\left(\textstyle{{m_0 \over 2}},
L_k, I(\delta_2),x,y\right) \right\}\ge 1 -\mathrm{e}^{-L_{k}^{\zeta_2}}
\end{equation}
for all $k=0,1,2,\dots$  and  $x, y \in \mathbb{Z}^d$ with 
$|x-y| > L_k +\varrho$.
\end{theorem}

The equivalent to (\ref{p111th}) holds in the context of
Theorem~\ref{thmmsa-subexp}, but it will not be needed.
In order to get sub-exponential decay of probabilities, the proof of 
Theorem~\ref{thmmsa-subexp}  allows the
number of bad boxes to grow with the scale.

\begin{proof}[Outline of the proof of Theorem~\ref{tBMSA}]

Theorem~\ref{tBMSA} is proven by a bootstrapping argument,
 making successive use of
 Theorems~\ref{thmmsa1},
\ref{msa2}, \ref{thmsuit-subexp}, and
\ref{thmmsa-subexp}.  We give here an  outline of the proof, and 
refer to \cite{GK1} for the full proof.
\begin{enumerate}
\item  Under the hypotheses  of Theorem~\ref{tBMSA}, 
we note that  hypothesis (\ref{hypH2})  of  Theorem~\ref{thmmsa1}
is the same as hypothesis (\ref{hypH})  for appropriate choices
of the parameters. 

\item We   apply Theorem~\ref{thmmsa1} obtaining
a sequence of length scales satisfying conclusion (\ref{x11}), with
its polynomial decay estimate of the probability of bad events. 

\item   In view of Remark \ref{regsuit},  it  follows that hypothesis
 (\ref{p1})  of  Theorem~\ref{msa2} is now satisfied at  suitably 
large scale. (We have bootstrapped from hypothesis (\ref{hypH}) 
to  hypothesis (\ref{p1})!).  Thus we can  apply Theorem~\ref{msa2} with
appropriate parameters, getting 
 $\delta_1 >0$ and a sequence of length scales satisfying
 conclusion (\ref{p111th}) for all $E \in I(\delta_1)$.  We set
$\delta_0=\delta_1$.

\item   We fix  $\zeta$ and $\alpha$  as in Theorem~\ref{tBMSA}, 
and pick $\zeta_0,\zeta_1,\zeta_2$ such that
$0<\zeta<\zeta_2<\zeta_1<\zeta_0<1<\alpha< \zeta_0\zeta_1^{-1}
 < \zeta_2^{-1}<\zeta^{-1} $. We note that we have bootstrapped  again: 
 hypothesis
 (\ref{hyp1_subexp}) of
 Theorem~\ref{thmsuit-subexp} is  satisfied at  all energies
 $E \in I(\delta_0)$ at appropriately large scale 
(the same for all $E$).
Applying Theorem~\ref{thmsuit-subexp}, we obtain a
 sequence of length scales for which
 conclusion (\ref{xxk_subexp}) holds for all $E \in I(\delta_0)$,
with its sub-exponential decay estimate of the probability of bad events. 

\item Using the last part of Remark \ref{regsuit}, we can see that we have 
bootstrapped to Theorem~\ref{thmmsa-subexp}: for any
 $0<\zeta_2< \zeta_1<\zeta_0<1$,
 hypothesis (\ref{p1_subexp}) is satisfied at all energies
 $E \in I(\delta_1)$ at sufficiently large scale
(depending on $\zeta_0,\zeta_1,\zeta_2 $ but independent of $E$).  We apply 
 Theorem~\ref{thmmsa-subexp}, obtaining $\delta_2>0$ and and
an exponentially growing sequence of length scales, depending on
  $\zeta_0,\zeta_1,\zeta_2 $, but independent of $E$, such that conclusion 
(\ref{msres_subexp}) holds for all  $E \in I(\delta_1)$.

\item  We have constructed in Step 5 
 a  sequence of
 length scales for which
(\ref{msres_subexp}) holds for all  $E \in I(\delta_0)$.
  Since the interval $I(\delta_0)$  (which is independent of $\zeta$)
 can be covered by 
$[\frac{\delta_0}{\delta_2}] +1$
 closed  intervals of length $\delta_2$,
we note that the desired conclusion (\ref{msres-bstp}) now follows from
(\ref{msres_subexp}), at the energies that are the centers
of the $[\frac{\delta_1}{\delta_2}] +1$ covering intervals, if we take $L_0$
appropriately large.

\end{enumerate}
\end{proof}

We will illustrate how to do a multiscale analysis by proving 
Theorem~\ref{msa2}
in Section~\ref{proofMSA}, and  refer to \cite{GK1}
for the proofs of Theorems~\ref{thmmsa1},
 \ref{thmsuit-subexp}, and
\ref{thmmsa-subexp}.

\section{From the multiscale analysis to localization} \label{secMASAloc}

The connection between the multiscale analysis and
 localization is given by the following theorem.

\begin{theorem} \label{tBMSA2}  
Let $H_\omega$ be a standard  ergodic random operator with (IAD)
and properties  (SGEE) and
 (EDI) in an open interval $\mathcal{I}$.  Then
\begin{equation}
\Sigma_{\mathrm{MSA}} \cap \mathcal{I} \subset
\Sigma_{\mathrm{EL}}  \cap \Sigma_{\mathrm{SSEHSKD}}
 \cap \mathcal{I} \, .
\end{equation}
\end{theorem}

To prove  Theorem~\ref{tBMSA2} we divide it into
 Theorems~\ref{tBMSA77} and \ref{tBMSA25}.  
Without loss of generality we assume that if 
 properties (GEE), (SGEE), or (EDI) hold,  then they hold for every
$\omega \in \Omega$.

\begin{lemma}  \label{lemmasing}
Let $H_\omega$ be a standard ergodic random operator
with properties (GEE) and  (EDI) in  an open interval $\mathcal{I}$.
Let us fix $m>0$.
For  every $\omega$,
given  $x \in \mathbb{Z}^d$  such that  
 there exists 
a generalized eigenfunction $\psi$ for
$H_\omega$ with generalized eigenvalue $E\in \mathcal{I}$ and
$\| \chi_{x}\psi\| \not= 0$,
 there exists
 $\widetilde{L}(\omega,E,m,x) < \infty$, such that 
the box  $\Lambda_{L}(x) $ is not 
 $(\omega,m, E)$-regular if $L \ge \widetilde{L}(\omega,E,m,x)$.
\end{lemma}

\begin{proof}  If  $x \in \mathbb{Z}^d$ and $\psi$ is
a generalized eigenfunction for
$H_\omega$ with generalized eigenvalue $E $,
 and the box  $\Lambda_{L}(x) $ is 
 $(\omega,m, E)$-regular, it follows from
(\ref{ediNew}) that for 
$E\in \mathcal{I}$ we have
\begin{eqnarray} \nonumber
\| \chi_x \psi \| &\le& \tilde{\gamma}_{\{E\}} \,
\mathrm{e}^{-m \frac L 2} 
\|  \Gamma_{x,L} \psi\| \le \tilde{\gamma}_{\{E\}} \,
\mathrm{e}^{-m \frac L 2} 
\| \langle x\rangle^{2\nu} \Gamma_{x,L}\|_{\infty} 
\|\psi\|_{\mathcal{H}_-} \\
&\le&
 d \tilde{\gamma}_{\{E\}} \|\psi\|_{\mathcal{H}_-} 
 L^{d-1}
\mathrm{e}^{-m \frac L 2} 
\langle |x|+\textstyle{ \frac L 2}-1\rangle^{2\nu}\nonumber \\
&\le&
4^\nu d \tilde{\gamma}_{\{E\}} \|\psi\|_{\mathcal{H}_-} 
\langle x\rangle^{2\nu}
 L^{d-1}
\langle \textstyle{ \frac L 2}-1\rangle^{2\nu}
\mathrm{e}^{-m \frac L 2} 
\, .\label{ediNew66}
\end{eqnarray}
Since the last expression in \eqref {ediNew66} goes to $0$
 as $L \to \infty$, the lemma follows. 
\end{proof}

The connection between the multiscale analysis and the
generalized eigenfunction expansion is given by the following
 lemma \cite[Lemma 4.1]{GK1}.

\begin{lemma}
\label{lemdecay}
Let $H_\omega$ be a standard ergodic random operator
with properties  (GEE) and  (EDI) in  an open interval $\mathcal{I}$. 
 Given an open interval $I$ with compact  $\bar{I} \subset \mathcal{I}$, $m>0$,
 $L \in 6\mathbb{N}$, and
 $x,y \in \mathbb{Z}^d$, let $ R(m,L, I,x,y) $ be as in  (\ref{eq:defsetR}).
For $\mathbb{P}$-almost every  $\omega\in R(m,L, I,x,y) $, we have
\begin{equation}\label{eq:wule}
\left\|\chi_{x}P_\omega(\lambda)\chi_{y}\right\|_2
\leq 
C  \tilde{\gamma}_{\bar{I}} \,\mathrm{e}^{-m \frac L 4} 
\langle x \rangle^{2\nu} \langle y \rangle^{2\nu},
\end{equation}
 for  $\mu_\omega$-almost all
 $\lambda\in I$, with $C=C(m,d,\nu )<+\infty$.
\end{lemma}

\begin{proof}
It follows from (\ref{pdagger}) that
$$
\left\|\chi_{x}P_\omega(\lambda)\chi_{y}\right\|_2
=
\left\|\chi_{y}P_\omega(\lambda)\chi_{x}\right\|_2\, ,
$$
for $\mu_\omega$-almost every $\lambda$,
so the roles played by $x$ and 
$y$ are symmetric.

Let $\omega\in R(m,L, I,x,y) $; then for any $\lambda\in I$, either 
$\Lambda_L(x)$
or $\Lambda_L(y)$ is $(m,\lambda)$-regular for $H_\omega$,
 say $\Lambda_L(x)$. If
$\phi\in\mathcal{H}$, for $\mu_\omega$-almost all 
$\lambda$  and all $y \in \mathbb{Z}^d$ the vector
$P_\omega(\lambda)\chi_{y}\phi$ is a  generalized eigenfunction of
$H_\omega$ with generalized eigenvalue $\lambda$, so 
 for $\mathbb{P}$-almost every $\omega$ it follows from 
property (EDI)
(see  (\ref{ediNew})), using $\chi_x=  \chi_{x, \frac L 3}\, \chi_x$, that
\begin{equation}
\| \chi_x P_\omega(\lambda)\chi_{y}\phi \| \le  
 \tilde{\gamma}_{\bar{I}} 
\| \Gamma_{x,L} R_{\omega,x,L}(\lambda) \chi_{x,\frac L 3}\| 
\|  \Gamma_{x,L}  P_\omega(\lambda)\chi_{y}\phi\| .\label{eigdec2}
\end{equation}
 Since 
 $\Lambda_L(x)$ is $(m,\lambda)$-regular, we have, using also
Lemma~\ref{lempoly} and the definition of the Hilbert-Schmidt  norm, that
\begin{eqnarray}
\| \chi_x P_\omega(\lambda)\chi_{y} \|_2  &\le & 
 \tilde{\gamma}_{\bar{I}} 
\mathrm{e}^{-m \frac L 2} \|  \Gamma_{x,L}
P_\omega(\lambda)\chi_{y}\|_2  \\ &\le &
 \tilde{\gamma}_{\bar{I}} 
d \left(\textstyle{ \frac 3 2}\right)^{2\nu}  L^{d-1}
\mathrm{e}^{-m \frac L 2} 
\langle |x|+\textstyle{ \frac L 2}-1\rangle^{2\nu}
\langle y \rangle^{2\nu} \\ &\le &
 \tilde{\gamma}_{\bar{I}} 
d 3^{2\nu} L^{d-1}
\mathrm{e}^{-m \frac L 2} 
\langle \textstyle{ \frac L 2}-1\rangle^{2\nu}
\langle x\rangle^{2\nu}
\langle y \rangle^{2\nu} \, ,
\end{eqnarray}
so \eqref{eq:wule} follows.
\end{proof}

\begin{theorem} \label{tBMSA77}  Let $H_\omega$ be a 
 standard ergodic random operator  with (IAD) and properties 
(GEE) and (EDI) in an open interval $\mathcal{I}$.  Then
\begin{equation}
\Sigma_{\mathrm{MSA}} \cap \mathcal{I} \subset
\Sigma_{\mathrm{EL}}  
 \cap \mathcal{I} \, .
\end{equation}
Moreover if  $E \in \Sigma_{\mathrm{MSA}} \cap \mathcal{I}$,
and we pick  an open interval $I \ni E$ and $m>0$ as in 
Definition~\ref{sigmaSMA} with compact $\bar{I} \subset   \mathcal{I}$,
then for $\mathbb{P}$-almost every  
$\omega$,  given  a generalized
 eigenfunction $\Psi$  for $H_\omega$
with generalized eigenvalue $E^\prime\in I$, we have
\begin{equation}\label{decayeigenfunction}
\limsup_{|x| \to \infty} \frac {\log\| \chi_x \psi \| }{|x|} \le 
- m \, .
\end{equation}
\end{theorem}

\begin{proof}  Given $E \in \Sigma_{\mathrm{MSA}} \cap \mathcal{I}$,
we pick an open interval $I\ni E$ as in Definition~\ref{sigmaSMA}
with compact $\bar{I} \subset   \mathcal{I}$. 
We fix  $\zeta$ and $\alpha$ such that $0<\zeta<1$ and
$1< \alpha < \zeta^{-1}$. 
By Definition~\ref{sigmaSMA} there is a scale $L_0$ and a 
mass $m>0$, such that, if
we set
$L_{k+1} = [L_k^\alpha]_{6\mathbb{N}}$, $k=0,1,\ldots$, then 
for $x$
and   $y \in \mathbb{Z}^d$ with $ |x-y|> L_k +\varrho$
we have the estimate
(\ref{msres-bstp})  for $k=0,1,2,\ldots$.

We will prove that $E \in \Sigma_{\mathrm{EL}}$
by showing that 
 for $\mathbb{P}$-almost every  $\omega$ each 
generalized eigenfunction of $H_\omega$ with 
generalized eigenvalue in $I$ is exponentially decaying
 in the $L^2$-sense.  This suffices since
  for $\mathbb{P}$-almost every  
$\omega$  we have that $\mu_\omega$-almost every  
$E^\prime \in \mathcal{I}$ is a generalized eigenvalue for
$H_\omega$, so we can then conclude that $H_\omega$ has pure point
spectrum in $I$.

 We fix $b >1$, to be chosen later.
Given  $x_0 \in  \mathbb{Z}^d$,
 for each  $k = 0,1,\cdots $
we define the discrete annulus
\begin{equation}
A_{k+1}(x_0) =\left\{
 \Lambda_{2bL_{k+1}}(x_0)\setminus \Lambda_{2L_k}(x_0) \right\}
\cap \mathbb{Z}^d \, ,
\end{equation}
 and  the event
\begin{eqnarray}\nonumber
\lefteqn{E_k(x_0) = 
\left\{\mbox{$\omega$;
$ \Lambda_{L_k}(x_0)$ and $\Lambda_{L_k}(x)$
 are both not
$(\omega,m,E^\prime)$-regular}\right.}\\
&& \hspace{1.6in}\left.\mbox{ for some  $E^\prime \in I $ and 
$x \in A_{k+1}(x_0)  $}\right\} \, .
\end{eqnarray}
By (\ref{msres-bstp}),
\begin{equation}
\mathbb{P}\left\{E_k(x_0)\right\} \le
 (2bL_{k+1})^d\,\mathrm{e}^{-2L_{k}^\zeta} \, , 
\end{equation}
and hence
\begin{equation}
 \sum^{\infty}_{k=0} \, \mathbb{P}\left\{E_k(x_0)\right\}
< \infty ,
\end{equation}
so it follows from the Borel-Cantelli Lemma and the countability of
$\mathbb{Z}^d$ 
 that
\begin{equation}
{\mathbb{P}}\{\mbox{$E_k(x_0)$ occurs infinitely often for some  $x_0 \in
{\mathbb{Z}}^d$  } \} = 0 \, .
\end{equation}
Thus,  for $\mathbb{P}$-almost every  
$\omega$, given  $x_0 \in{\mathbb{Z}}^d$ there is
$k_1(\omega,x_0) \in \mathbb{N}$ such that  
$\omega \notin  E_k(x_0)$ for $k \ge k_1(\omega,x_0)$.

For $\mathbb{P}$-almost every  
$\omega$,  given  a generalized
 eigenfunction $\Psi$  for $H_\omega$
with generalized eigenvalue $E^\prime\in I$, we
pick $x_0 \in \mathbb{Z}^d$ such that $\| \chi_{x_0}\psi\| \not= 0$.
We  set 
$k_2(\omega, E^\prime,x_0) =\min\{k\in\mathbb{N};
 L_k \ge \widetilde{L}(\omega,E^\prime,m,x_0)\}$,
where  $ \widetilde{L}(\omega,E^\prime,m,x_0)$ is as in
  Lemma~\ref{lemmasing}.  Thus, if 
$k_3(\omega, E^\prime,x_0)= 
\max \{k_1(\omega,x_0),k_2(\omega, E^\prime,x_0)\}$,
for  $k \ge k_3(\omega,  E^\prime,x_0)$ we conclude that 
$\Lambda_{L_k}(x)$ is
$(\omega,m,E^\prime)$-regular for all $x \in A_{k+1}(x_0)$. 
We pick  $\rho$, with $\frac 13 < \rho < 1$,   and $b > {1+\rho \over
1-\rho}$, and set
\begin{equation}
\tilde A_{k+1}(x_0) = \left\{\Lambda_{{2b\over 1+\rho}L_{k+1}}(x_0)\setminus
\Lambda_{{2\over 1-\rho}L_k}(x_0)\right\}
\cap \mathbb{Z}^d \, ,
\end{equation}
Note that $\tilde{A}_{k+1}(x_0) \subset A_{k+1}(x_0)$ and
\begin{equation} \label{dist}
\mathrm{dist} (x,\mathbb{Z}^d \backslash A_{k+1}(x_0)) \ge 
\rho|x-x_0| \;\;\mbox{for all} \;\;  x \in
\tilde A_{k+1}(x_0)\,  .
\end{equation}
Thus, if  $x \in A_{k+1}(x_0)$ with $k \ge k_3(\omega, I,x_0)$,
it follows from (\ref{edi})  that
\begin{eqnarray} \label{edi77}
\| \chi_x \psi \| \le  d\tilde{\gamma}_{\bar{I}}  L_k^{d-1}\,
\mathrm{e}^{-m \frac {L_k} 2} \| \chi_{x_1} \psi \|
\end{eqnarray}
for some $x_1 \in \Upsilon_{L}(x)$.  If we take 
 $x \in \tilde{A}_{k+1}(x_0)$, we have
 $x_1 \in  A_{k+1}(x_0)$ in view of  \eqref{dist}, and hence
we can apply again (\ref{edi}) as in  (\ref{edi77})
to estimate $\| \chi_{x_1} \psi \|$ in terms of some 
$\| \chi_{x_2} \psi \|$ for some $x_2 \in \Upsilon_{L}(x_1)$.
In fact, it follows from  \eqref{dist} that
for  $x \in \tilde{A}_{k+1}(x_0)$  this procedure can be repeated
$n$ times, yielding
\begin{eqnarray} \label{edi778}
\| \chi_x \psi \| &\le & \left(d\tilde{\gamma}_{\bar{I}}  L_k^{d-1}\,
\mathrm{e}^{-m \frac {L_k} 2}\right)^n \| \chi_{x_n} \psi \|  \\
&\le& \textstyle{\left(\frac 3 2\right)^\nu} \|\psi\|_{\mathcal{H}_-}
\left(d\tilde{\gamma}_{\bar{I}}  L_k^{d-1}\,
\mathrm{e}^{-m \frac {L_k} 2}\right)^n \langle x_n\rangle^{2\nu}
\label{edi779}
\end{eqnarray}
for some $x_n \in \mathbb{Z}^d$ with $|x_n -x| \le n (\frac{L_k}2 -1)$,
as long as
$ n(\frac{L_k}2 -1)  < \rho|x-x_0|$. (We used \eqref {eq:normchi}
to obtain \eqref{edi779}).  We thus have the estimate  \eqref{edi779})
with
\begin{eqnarray}  \label{n}
n= \frac{ \rho|x-x_0|}{\frac{L_k}2 -1} - 1 
\ge  \frac{ \textstyle{\frac {3 \rho -1}2}|x-x_0|}{\frac{L_k}2 -1} \, .
\end{eqnarray}
Note that for all $k$ sufficiently large we have $\frac{L_k}2 -1 \ge \frac{L_k}4$
and $d\tilde{\gamma}_{\bar{I}}  L_k^{d-1}\,
\mathrm{e}^{-m \frac {L_k} 2} \le \mathrm{e}^{-m \frac {L_k} 4}$, 
in which case  it follows from \eqref{edi778} and \eqref{n}
that for each $x \in \tilde{A}_{k+1}(x_0)$ we have
\begin{eqnarray}
\| \chi_x \psi \| & \le&
 \textstyle{\left(\frac 3 2\right)^\nu} \|\psi\|_{\mathcal{H}_-}
 \langle |x_0| + \rho |x-x_0|\rangle^{2\nu}
\,\mathrm{e}^{-  \textstyle{\frac {3 \rho -1}2} m |x-x_0|}\\
& \le& 3^\nu  \|\psi\|_{\mathcal{H}_-} \langle x_0\rangle^{2\nu}
 \langle\rho |x-x_0|\rangle^{2\nu}
\, \mathrm{e}^{-  \textstyle{\frac {3 \rho -1}2} m |x-x_0|} \, .
\end{eqnarray}
Thus there exists $\tilde{k}$, depending only on $\rho$, $d$, 
$\nu$, $ \|\psi\|_{\mathcal{H}_-} $, $x_0$, $L_0$, $\alpha$, 
$\tilde{\gamma}_{\bar{I}}$, and $m$, such that if
 $x \in \tilde{A}_{k+1}(x_0)$ with $k \ge \tilde{k}$ we have (recall
$\frac 13< \rho < 1$)
\begin{equation}
\| \chi_x \psi \|  \le 
\mathrm{e}^{-  \textstyle{\frac {\rho(3 \rho -1)}2} m |x-x_0|}\, .
\end{equation}
 Since if $x \in \mathbb{Z}^d$ is such that
 $|x-x_0| > {L_0\over 1-\rho}$, we have  $x
\in \tilde A_{k+1}(x_0)$ for some $k$, we conclude that there is a finite constant
$C_{\psi,\rho}$ such that
\begin{equation}
\| \chi_x \psi \|  \le C_{\psi,\rho}\, 
\mathrm{e}^{-  \textstyle{\frac {\rho(3 \rho -1)}2} m |x-x_0|}
\;\;\;\mbox{for all $x \in \mathbb{Z}^d$}\, ,
\end{equation}
and hence $\psi$ decays exponentially in the $L^2$-sense.  In fact, 
we proved that for each $\frac 13< \rho < 1$ we have
\begin{equation}
\limsup_{|x| \to \infty} \frac {\log\| \chi_x \psi \| }{|x|} \le
-\textstyle{\frac {\rho(3 \rho -1)}2} m \, ,
\end{equation}
so letting $\rho \to 1$ we get \eqref{decayeigenfunction}.
\end{proof}

We  now show that the multiscale analysis imply
 strong sub-exponential
HS-kernel decay \cite[Theorem 3.8]{GK1}.  (Note that for smooth
functions of Schr\"odinger and classical wave operators
 we always have kernel decay in the deterministic case \cite{GK2,BGK}.)

\begin{theorem} \label{tBMSA25}  Let $H_\omega$ be a 
 standard ergodic random operator  with (IAD) and properties 
(SGEE) and (EDI) in an open interval $\mathcal{I}$.  Then
\begin{equation}
\Sigma_{\mathrm{MSA}} \cap \mathcal{I} \subset
\Sigma_{\mathrm{SSEHSDC}}
 \cap \mathcal{I} \, .
\end{equation}
\end{theorem}

\begin{proof}
Given $E \in \Sigma_{\mathrm{MSA}} \cap \mathcal{I}$,
we pick an open interval $I\ni E$ as in Definition~\ref{sigmaSMA}
with compact $\bar{I} \subset   \mathcal{I}$. 
 We will use  the
generalized eigenfunction expansion (\ref{eq:expantrace})  to show
that for any $0<\xi<1$. there is a finite constant $C_\xi$ such that 
\begin{equation} \label{expec}
\mathbb{E}\left\{\sup_{|\!|\!|f|\!|\!|\leq 1} \left\|\chi_{x}
f(H_\omega)E_{\omega}(I)
\chi_{0}\right\|_2^2 \right\}
\leq  C_\xi\, \mathrm{e}^{-|x|^\xi},
\end{equation} 
for all $ x\in \mathbb{Z}^d$,
the supremum being taken
over all  Borel functions $f$ of a real variable, with 
$\!|\!|\!| f |\!|\!|=\sup_{t \in \mathbb{R}} |f(t)|$.
 Since our random operator is 
 $\mathbb{Z}^d$-ergodic,  probabilities are
translation invariant, so there is no loss of generality in taking $y=0$.

Given $0<\xi<1$, we pick $\zeta$ such that
$\zeta^2 <\xi<\zeta <1$  (always possible) and set $\alpha = \frac \zeta \xi $,
note $\alpha < \zeta^{-1}$. 
By Definition~\ref{sigmaSMA} there is a scale $L_0$ and a 
mass $m_\zeta>0$, such that, if
we set
$L_{k+1} = [L_k^\alpha]_{6\mathbb{N}}$, $k=0,1,\ldots$, then for each $k$ 
we have the estimate
(\ref{msres-bstp}) with $y=0$
and   $x \in \mathbb{Z}^d$ such that $ |x|> L_k +\varrho$.

Let us now fix $x \in \mathbb{Z}^d$ and pick $k$ such that
 $L_{k+1} +\varrho \geq |x|> L_k +\varrho$.  In this case
Lemma~\ref{lemdecay}
asserts that if $\omega\in R\left(m_\zeta, L_k, I,x,0\right)$, then
\begin{eqnarray} \label{eqdecay}
 \left\|\chi_{x} P_\omega(\lambda)\chi_{0}
\right\|_2
\leq C_1 \,\mathrm{e}^{-m_\zeta \frac {L_k}4} 
\langle x \rangle^{2\nu}
\le C_1 C_2\,\mathrm{e}^{-  L_k^\zeta}  \ ,
\end{eqnarray}
 for  $\mu_\omega$-almost all
 $\lambda\in I$,
with finite constants 
$C_1=C_1(m_\zeta,d,\nu,\tilde{\gamma}_{\bar{I}})$
  and $C_2=C_2(\nu,\varrho, \zeta,\xi,  m_\zeta) $.
We split the expectation in (\ref{expec}) in two pieces:  where
(\ref{eqdecay}) holds, and over the complementary event,
which has probability less than 
$\mathrm{e}^{-L_{k}^{{\zeta}}}$ by (\ref{msres-bstp}).
From (\ref{eq:expantrace}) we have  
\begin{eqnarray}
\lefteqn{ \hspace{-2cm}
\sup_{|\!|\!|f|\!|\!|\leq 1} 
\left\|\chi_{x}f(H_\omega)E_{\omega}(I) 
\chi_{0}\right\|_2 } \nonumber \\
& \leq &
\sup_{|\!|\!|f|\!|\!|\leq 1}
\int_{I} |f(\lambda)| \left\|\chi_{x} 
P_\omega(\lambda)\chi_{0} \right\|_2
{\rm d}\mu_\omega(\lambda) \label{ineqop} \\
& \leq &
\int_{I} \left\|\chi_{x} 
P_\omega(\lambda)\chi_{0} \right\|_2
{\rm d}\mu_\omega(\lambda).\label{eq:stepop}
\end{eqnarray}
Thus, it follows from (\ref{eqdecay}) that 
[with $\mathbb{E} (F(\omega) ; A) \equiv 
\mathbb{E} (F(\omega)\chi_A(\omega)) $]
\begin{eqnarray}
\lefteqn{
\mathbb{E}\left\{ \sup_{|\!|\!|f|\!|\!|\leq 1}
\left\|\chi_{x} f(H_\omega)E_{\omega}(I)
\chi_{0}\right\|_2^2; 
{R(m_\zeta,L_k, I,x,0)} \right\} }\nonumber \\
&  & \hspace{1.5in} \leq
C_1^2C_2^2\;\mathbb{E}\{(\mu_\omega(I))^2\} 
\,\mathrm{e}^{-  2L_k^\zeta}\ .
\label{eq:stepxy1}
\end{eqnarray}
To estimate the second term,  note that using  (\ref{012}) we have
\begin{eqnarray}
\left\|\chi_{x}
f(H_\omega)E_{\omega}(I) \chi_{0}\right\|_2^2
& \le & |\!|\!|f|\!|\!|^2  \left\|E_{\omega}(I)
\chi_{0}\right\|_2^2 \nonumber \\
& \le & 4^\nu   |\!|\!|f|\!|\!|^2 \mu_\omega(I) \, ,
\end{eqnarray}
so, using the Schwarz's inequality and  (\ref{msres-bstp}) , 
\begin{eqnarray}
\lefteqn{
\mathbb{E}\left\{ \sup_{|\!|\!|f|\!|\!|\leq 1}
\left\|\chi_{x} f(H_\omega)E_{H_\omega}(I)
\chi_{0}\right\|_2^2; 
{\omega \notin R(m_\zeta,L_k, I,x,0)} \right\} }\hspace{1.0in}\nonumber \\
&  &  \hspace{.6in}\leq
4^\nu\;[\mathbb{E}\{(\mu_\omega(I)^2\}]^{\frac12} 
\,\mathrm{e}^{-   {\frac12} L_k^\zeta}\,.
 \label{eq:stepxy12}
\end{eqnarray}
Since
\begin{equation}
C_3=C_1^2C_2^2\;\mathbb{E}\{(\mu_\omega(I))^2\}  +
4^\nu\;[\mathbb{E}\{(\mu_\omega(I)^2\}]^{\frac12}  <\infty
\end{equation}
 in view of (\ref{eq:ingr1bis}), 
we   conclude from (\ref{eq:stepxy1}) and (\ref{eq:stepxy12})  that
 (recall $\alpha = \frac \zeta \xi $)
\begin{eqnarray} \label{expec2}
\lefteqn{\mathbb{E}\left\{\sup_{|\!|\!|f|\!|\!|\leq 1} \left\|\chi_{x}
f(H_\omega)E_{\omega}(I)
\chi_{0}\right\|_2^2\right\} }\\
&& \hspace{.5in} \le  C_5 \,\mathrm{e}^{- {\frac12} L_k^\zeta} 
\le C_3 \,\mathrm{e}^{-{\frac12} L_{k+1}^{\xi}} 
 \leq C_3 \mathrm{e}^{-{\frac12}(|x|- \varrho)^\xi} \le 
C_3 \mathrm{e}^{{\frac12}\varrho^\xi}\mathrm{e}^{-{\frac12} |x|^\xi} \,
\nonumber
\end{eqnarray}
for all $|x| \ge L_0 + \varrho$.
Thus  (\ref{expec}) follows (for a slightly smaller $\xi$), and
Theorem~\ref{tBMSA25} is proved.
\end{proof}

\section{How to do  a multiscale analysis}\label{proofMSA}

To exemplify how to perform  a multiscale analysis we  give the 
proof of Theorem~\ref{msa2},
a modification of
 the  proof of \cite[Theorem 2.2]{VDK}.

\begin{proof}[Proof of Theorem~\ref{msa2}]

Given $x \in \mathbb{Z}^d$ we set
\begin{eqnarray} \label{xiL}
\Xi_{L,\ell}(x)& = & 
\Lambda_{L}(x) \cap \left\{x+\frac{\ell}{3} {\mathbb{Z}}^d\right\}
 \subset {\mathbb{Z}}^d  \ ,
 \; \;  \Xi_{L,\ell}= \Xi_{L,\ell}(0) \, ,\\
{\mathcal{C}}_{L,\ell}(x)& =&\{ \Lambda_{\ell}(y); \;\;y \in \Xi_{L,\ell}(x)\, ,  
\Lambda_{\ell}(y)\sqsubset\Lambda_{L}(x) \} 
\, , \; \; {\mathcal{C}}_{L,\ell}={\mathcal{C}}_{L,\ell}(0)\,.
\end{eqnarray}
Note
$| \Xi_{L,\ell}(x)| \le (3 \frac L \ell+1)^d$.
By a \emph{cell} we will  mean a closed box 
 $\overline{\Lambda}_{\ell/3}(y)$ with $ y\in\Xi_{L,\ell}(x)$,  
 the \emph{core} of the box ${\Lambda}_{\ell}(y)$.
Thus ${\mathcal{C}}_{L,\ell}(x)$ is the collection of boxes of side $\ell$
whose core is a cell and are inside the boundary belt 
 $\tilde{\Upsilon}_L(x)$ of the big  box ${\Lambda}_{L}(x)$; 
we have  $| {\mathcal{C}}_{L,\ell}(x)| \le  (3 \frac L \ell-2)^d$.
Note that the big box is  covered by  cells:  
$\overline{\Lambda}_{L}(x) \subset
\bigcup_{y \in\Xi_{L,\ell}(x)} \overline{\Lambda}_{\ell/3}(y)$. 

Given $\theta, p, p^\prime$ such that
\begin{equation} \label{thetap}
 0< {p} <p^{\prime} < \theta - bd\; \;\;\mbox{and} \;\;\;
  1<\alpha<\min\left\{\frac {2{p}+2d}{{p} + 2d}, 
\frac{\theta }{{p} +bd}\right\}\, ,
\end{equation}
we pick  $s$ and $\theta^\prime$ such that
\begin{equation} \label{thetap2}
\frac \theta 2<  \theta^\prime\; \;\;\mbox{and} \;\;\;
 {p}+ bd  < s<\alpha s< \theta^\prime < \theta \, .
\end{equation}
 Recalling 
$m_0= 2\theta {\log L_0 \over L_0}$, we have
\begin{equation}
\frac{m_0}{2} <m_0^\prime= 2\theta^\prime {\log L_0 \over L_0} < m_0\, .
\end{equation}

If $\Lambda_{L_0}(x)$  is $(\omega, m_0, E_0)$-regular and
$\mathrm{dist} (\sigma (H_{\omega, x,L_0}),E_0)>{L_0^{-s}}$,
 it  follows from the (first) resolvent identity that
 $\Lambda_{L_0}(x)$  is $(\omega, {m_0^\prime}, E)$-regular for all
$E \in I=[E_0-\delta, E_0+\delta]\cap I_0$, where
\begin{equation} \label{delta100}
\delta=\delta (\theta,\theta^\prime,s,L_0) = {1 \over {2L_0^{2s}}} 
\left( \mathrm{e}^{-{m_0^\prime}{L_0\over
2}} - \mathrm{e}^{-m_0{L_0\over  2}}\right) \, .
\end{equation}
 Using the hypothesis \eqref{p1} with Remark~\ref{regsuit},
plus property (W)  at $E_0$ with $\eta = L_0^{-s}$ (see \eqref{wegner}) ,
we conclude that
\begin{eqnarray} \label{stat222}
\lefteqn{\mathbb{P} \{\mbox{
$\Lambda_{L_0}(0)$ is 
$(\omega,{m_0^\prime}, E)$-regular for every $E \in I $}  \} }   \\
 &  &\hspace{1.4in} \ge 
1 -\frac1{L_0^{p^{\prime}}} - \frac{Q_{I_0}}{ L_0^{s-bd}} \ge 1 -\frac1{L_0^{{p}}}\
 \nonumber
\end{eqnarray}
if $L_0 \ge \mathcal{B}_1=  \mathcal{B}_1(d,b,Q_{I_0},p,p^{\prime},s)$.
  Combining with
property  (IAD), we get that for $L_0 \ge \mathcal{B}_1$ we also have
\begin{equation} \label{start1}
\mathbb{P}\left\{R\left(m_0^\prime, L_0,
I,x,y\right) \right\} 
\ge 1 -\frac1{L_0^{2{p}}}
\end{equation}
 for all $x, y \in \mathbb{Z}^d$ with 
$|x-y| > L_0 +\varrho$.

We will first prove the  weaker conclusion (\ref{p111th}) by
 a single energy multiscale analysis which is basically the multiscale 
analysis of von Dreifus \cite{VD}, except that singular regions 
 are treated as in \cite{VDK}.
Let us fix $E \in I$, it obviously follows from \eqref{stat222} that
\begin{eqnarray} \label{stat2229}
\mathbb{P} \{
\Lambda_{L_0}(0) \;\;\mbox{is 
$(\omega,{m_0^\prime}, E)$-regular}  \}   \ge 1 -\frac1{L_0^{{p}}}\
\end{eqnarray}
if $L_0 \ge \mathcal{B}_1$.   Conclusion (\ref{p111th})  is  proven by induction. 
 Given a scale
$L \in6\mathbb{N}$ and  $m>0$, we let $p_L(m)$ be the probability 
that a  box at scale $L$ is $(\omega,{m}, E)$-singular 
(not $(\omega,{m}, E)$-regular),  i.e.,
\begin{eqnarray}  \label{pL1}
p_L(m) =
\mathbb{P}\{\Lambda_{L}(0) \;\mbox{ is  $(\omega,m, E)$-singular} \} \, .
\end{eqnarray}
 The induction step
 goes from scale $\ell \ge L_0$ to scale 
$L= [\ell^\alpha]_{6\mathbb{N}}\,$: given 
\begin{equation}
p_\ell(m) < \frac1{\ell^{{p}}}\;\;\; \mbox{with}\;\;\; 
 m=m_\ell \ge  2\theta^\prime {\log \ell \over\ell}\, , 
\end{equation}
we prove
\begin{equation}
p_L(M) <\frac1{L^{{p}}} 
\;\;\; \mbox{for some}\;\;\;  
 M= m _L \ge  2\theta^\prime {\log L \over L}\, .
\end{equation}
 To finish the proof of
 (\ref{p111th}), we   show $\inf_{k} m_{L_k} \ge \frac{m_0}2$, i.e.,
\begin{eqnarray}\label{summ}
 \sum_{k=0}^\infty (m_{L_k} -m_{L_{k+1}}) \le m_0^\prime - \frac{m_0}2 \, .
\end{eqnarray}

 The induction step proceeds roughly as in \cite{VDK}.  The deterministic
part is based on the SLI,
 but only boxes in ${\mathcal{C}}_{L,\ell} $ are  allowed.  
 The basic idea is that
if all boxes in ${\mathcal{C}}_{L,\ell}$ were $(\omega,{m}, E)$-regular,
 then it would
 follow from applying the estimate 
 (\ref{sli}) repeatedly that the big box $\Lambda_L (0)$ is also 
$(\omega,{M}, E)$-regular with the difference  $m-M$ ``small".

To see how this works, for a given $x \in \mathbb{Z}^d$
 we fix ${x_0\in\Xi_{\frac{L +\ell}3,\ell}}(x)$ and apply the SLI estimate
 (\ref{sli}) repeatedly with $\ell^\prime=\ell$, as long as we do not
 hit the boundary belt $\tilde{\Upsilon}_{L}(x)$ (see \eqref{bdrybelt}).  
Each time  the SLI is performed one gains a
 factor of $3^d \gamma_{I}$  and moves to
 an  adjacent cell (see Remark~\ref{rsli}).  After $N$ applications we have
\begin{eqnarray} \label{sl300}
\lefteqn{ \| {\Gamma_{x,L}} R_{\omega,x,L}(E) \chi_{x_0,{\frac \ell 3}}\| }\\  && \le
\left(3^d\gamma_{I}\right)^N \Pi_{i=1}^N 
 \| \Gamma_{x_i,\ell} R_{\omega,x_i,\ell}(E) \chi_{x_i,{\frac \ell 3}}\| 
 \| {\Gamma_{x,L}} R_{\omega,x,L}(E) \chi_{x_N,{\frac \ell 3}}\|  \nonumber \, ,
\end{eqnarray}
where $x_0,x_1,\ldots,x_N \in \Xi_{{L},\ell}(x)$ are centers of adjacent cells 
which are cores of boxes in ${{\mathcal{C}}_{L,\ell}(x)}$, i.e., 
$|x_{i}-x_{i-1}| = \frac \ell 3$ and
$\Lambda_{ \ell }(x_i)\in{{\mathcal{C}}_{L,\ell}(x)}$ for 
$i=0,1,\ldots,N$.  A moment of reflection shows that we are always 
in this situation as long as
\begin{equation}
(N-1)\frac \ell 3 \le \frac{L -3 } 2 - \frac \ell 2 - \frac {L+\ell} 6 \, .
\end{equation}
Since $N$ is an integer, we can always take $N$ to be the 
unique integer satisfying
\begin{equation} \label{sl305}
\frac L \ell -3 < N \le \frac L \ell -2 \, .
\end{equation}
If all boxes in ${{\mathcal{C}}_{L,\ell}(x)}$ are $(\omega,{m}, E)$-regular
we conclude from \eqref{sl300} and \eqref{sl305} that
\begin{eqnarray} \label{sl390}
 \| {\Gamma_{x,L}} R_{\omega,x,L}(E) \chi_{x_0,{\frac \ell 3}}\| \le
\left(3^d\gamma_{I}\mathrm{e}^{-m{\ell\over  2}}\right)^{ \frac L \ell -3}
\| R_{\omega,x,L}(E) \|   \, .
\end{eqnarray}

Thus, 
\begin{eqnarray}
\| {\Gamma_{x,L}} R_{\omega,x,L}(E) \chi_{x,\frac L 3}\|   & \le  &
\sum_{x_0 \in\Xi_{\frac{L +\ell}3,\ell}(x)} \| {\Gamma_{x,L}} R_{\omega,x,L}(E)
\chi_{x_0,{\frac \ell 3}}\|\nonumber \\   & \le  &
\left({L\over \ell}+2\right)^d 
\sup_{x_0 \in\Xi_{\frac{L +\ell}3,\ell}(x)} \| {\Gamma_{x,L}}
R_{\omega,x,L}(E)
\chi_{x_0,{\frac \ell 3}}\|
\nonumber \\    & \le  &\left({L\over \ell}+2\right)^d
\left(3^d\gamma_{I}\mathrm{e}^{-m{\ell\over  2}}\right)^{ \frac L \ell -3}
\| R_{\omega,x,L}(E) \| .\label{bigboxest}
\end{eqnarray} 
If
$\| R_{\omega,x,L}(E) \| \le L^{s}$, which holds
outside   a set of small probability by the Wegner estimate \eqref{wegner},
 we get
\begin{eqnarray}
\| {\Gamma_{x,L}} R_{\omega,x,L}(E) \chi_{x,\frac L 3}\|   \le  L^{s}
\left({L\over \ell}+2\right)^d  
\left(3^d\gamma_{I}\mathrm{e}^{-m{\ell\over  2}}\right)^{ \frac L \ell -3}
 \equiv \mathrm{e}^{-M{L\over  2}}\, ,
\label{bigboxest3}
\end{eqnarray}
with\begin{eqnarray}
M \ge m\left(1 - \frac c {\log \ell}\right) \ge   2\theta^\prime {\log L \over L}
\end{eqnarray}
for $\ell$ sufficiently large, with $c$ a constant depending only on 
$d, \gamma_I,\theta^\prime,s,\alpha$ and $L_0$.
The desired estimate \eqref{summ} follows if  $L_0$ is large enough.

Unfortunately the probabilistic estimates do not work.  We assumed that 
 all boxes in ${{\mathcal{C}}_{L,\ell}(x)}$ are $(\omega,{m}, E)$-regular and
$\| R_{x,L}(E) \| \le L^{s}$, thus we can only conclude that
\begin{eqnarray}
p_L(M) &\le&  (3 \frac L \ell-2)^d p_\ell(m) + Q_I \frac 1 {L^{s -bd}}\le
 (3 \frac L \ell-2)^d \frac 1 {\ell^{{p}}} + Q_I  \frac 1 {L^{s -bd}}
\nonumber\\
&\le&  \frac 3 {\ell^{{p} - (\alpha -1)}} + 
Q_I  \frac 1{L^{s -bd}}\, .
\end{eqnarray}
To get $p_L(M) \le \frac 1 {L^{{p}}}$ we would need
${p} - (\alpha -1) > {p} $, which is impossible
since $\alpha >1$.

To fix this problem we must relax the condition that
 all boxes in ${{\mathcal{C}}_{L,\ell}(x)}$ are $(\omega,{m}, E)$-regular
and accept the presence of   at least one
 $(\omega,{m}, E)$-singular  box in 
${{\mathcal{C}}_{L,\ell}(x)}$. 
To exploit the independence of events in  
nonoverlapping boxes (property (IAD)) we will forbid
the existence of  two \emph{nonoverlapping}   singular boxes in 
${{\mathcal{C}}_{L,\ell}(x)}$.  

To see how we obtain the improvement in the probabilities, 
let
us consider the event
\begin{eqnarray} \label{Koverlap}
\lefteqn{Q_{x}^{(K)}(E,\ell,L,m) =}\\
&&\{\omega; \;\mbox{there are $K$ nonoverlapping  $(\omega,{m}, E)$-singular
 boxes in 
${{\mathcal{C}}_{L,\ell}(x)}$}\} .\nonumber
\end{eqnarray}
Using property (IAD) we get
\begin{eqnarray}\nonumber
\mathbb{P}\{Q_{x}^{(2)}(E,\ell,L,m)\}&\le&  
| {{\mathcal{C}}_{L,\ell}(x)}|^2p_\ell(m)^2 \le 
 \left(3 \frac L \ell-2\right)^{2d}\frac{1}{\ell^{2{p}}}\le  9^d  \ell^{2d(\alpha -1)} \frac{1}{\ell^{2{p}}}\\
&<& 
\frac{1}{2\ell^{\alpha {p}}}\le  \frac{1}{2L^{ {p}}} \, ,
\label{p999}
\end{eqnarray}
with \eqref{p999} valid for large $\ell$ if
 $\alpha < \frac {2{p}+2d}{{p} + 2d} = 1 + \frac {{p}}{{p} + 2d}$, which
allows for $\alpha >1$.

We may have fixed one problem but we created another: we cannot 
estimate the right hand side of  \eqref{sl300} as before,
because we may hit a singular box, i.e., some of the $x_i$'s in \eqref{sl300}
 may not be the centers of
$(\omega,{m}, E)$-regular boxes.  So we must make changes.  
Taking $\omega \notin Q_{0}^{(2)}(E,\ell,L,m)$ we exclude
the possibility of two nonoverlapping bad boxes in  ${{\mathcal{C}}_{L,\ell}(x)}$,
so if there is one singular box, say $\Lambda_\ell(u)$ (note $u$ depends on
 $\omega, \ell,m,E$), to guarantee
 that $\Lambda_\ell(u^\prime) \in {{\mathcal{C}}_{L,\ell}(x)}$ is
a regular box we need $|u ^\prime - u|  > \ell + \varrho$.  
Taking $\ell >3\varrho$,
it suffices to have  $|u ^\prime - u|  > \frac {4\ell}3 $.
Thus $\Lambda_{\frac {7\ell}3}(u)$ is our ``singular region", i.e., the region 
such that boxes in 
 ${{\mathcal{C}}_{L,\ell}(x)}$ with cores  outside this region are regular.
Given ${x\in\Xi_{\frac{L +\ell}3,\ell}}$, we estimate
$ \| {\Gamma_{x,L}} R_{\omega,x,L}(E) \chi_{x,{\frac \ell 3}}\| $ by applying
 the SLI estimate
 (\ref{sli}) repeatedly,  as long as we do not
 hit the boundary belt $\tilde{\Upsilon}_{L}(0)$, but we now have two cases:
\begin{itemize}
\item If  $x^\prime \notin \Lambda_{\frac {7\ell}3}(u)$ and 
 $\Lambda_\ell(x^\prime)  \in
{{\mathcal{C}}_{L,\ell}(x)}$,  then $x^\prime$
 is the center of a regular box in  ${{\mathcal{C}}_{L,\ell}(x)}$ and
we use  (\ref{sli}) with $\ell^\prime=\ell$, obtaining
\begin{eqnarray}\label{sligood}
\| {\Gamma_{x,L}} R_{\omega,x,L}(E) \chi_{x^\prime,{\frac \ell 3}}\|  \le
3^d \gamma_{I}   \mathrm{e}^{-m{\ell\over  2}}
\| {\Gamma_{x,L}} R_{\omega,x,L}(E) \chi_{x^{\prime\prime},{{\frac \ell 3}}}\|
\end{eqnarray}
for some 
$x^{\prime\prime} \in \Upsilon_{\ell,\ell}(x^{\prime})$,
 i.e., 
$|x^{\prime\prime}-x^{\prime}| = \frac \ell 3$.

\item
 If   $x^\prime \in \Lambda_{\frac {7\ell}3}(u)$
and $\Lambda_{{\frac {11\ell}3}}(u)\sqsubset \Lambda_L(x)$, we apply
the SLI estimate (\ref{sli})   with $y=x^\prime$, $y^\prime= u$, and
 $\ell^\prime= 3 \ell$, so $k=9$, obtaining
\begin{eqnarray} \label{slibad}
\lefteqn{\| {\Gamma_{x,L}} R_{\omega,x,L}(E) \chi_{x^\prime,{\frac \ell 3}}\| }\\
&& \le
9^d\gamma_{I}   
\| \Gamma_{u,3\ell} R_{\omega,u,3\ell}(E) \chi_{x^\prime,{\frac \ell 3}}\| 
\| {\Gamma_{x,L}} R_{\omega,x,L}(E) \chi_{x^{\prime\prime},{{\frac \ell 3}}}\|
\nonumber
\end{eqnarray} 
for some $x^{\prime\prime} \in \Upsilon_{3\ell,\ell}(u)$
(see \eqref{2Upsilon}), so 
$|x^{\prime\prime}-u| = \frac {4\ell} 3$, and hence
$x^{\prime\prime} \notin   \Lambda_{\frac {7\ell}3}(u)$ with
$\Lambda_\ell(x^{\prime\prime}) \in {{\mathcal{C}}_{L,\ell}(x)}$.
We are now in the previous case, so we can use   (\ref{sligood}) to get
\begin{eqnarray} \label{slibad2}
\lefteqn{\| {\Gamma_{x,L}} R_{\omega,x,L}(E) \chi_{x^\prime,{\frac \ell 3}}\| }\\
&& \le 
27^d \gamma_{I}^2   \mathrm{e}^{-m{\ell\over  2}}
\|  R_{\omega,u,3\ell}(E) \| 
\| {\Gamma_{x,L}} R_{\omega,x,L}(E) \chi_{x^{\prime\prime\prime},{{\frac \ell 3}}}\|
\nonumber
\end{eqnarray} 
for some
 $x^{\prime\prime\prime} \in \Upsilon_{\ell,\ell}(x^{\prime\prime})
$; note $|x^{\prime\prime\prime}-u| \le \frac {5\ell} 3$  
and $|x^{\prime\prime\prime}-x^\prime| \le \frac {8\ell} 3$.

\end{itemize}
To control $\|  R_{\omega,u,3\ell}(E) \| $ in \eqref{slibad2}
and $\|  R_{\omega,x,L}(E) \| $ in the final expression 
we will require
\begin{eqnarray}  \label{w1}
\|  R_{\omega,u,3\ell}(E) \| \le L^s \;\; \mbox{for all 
$u \in  \Xi_{L,\ell}(x)$}\, ,
\end{eqnarray}
and
\begin{eqnarray}
\|  R_{\omega,0,L}(E) \| \le L^s  \label{w2} \, .
\end{eqnarray}
To do so, let us define the events
\begin{eqnarray} \label{Wall}
\lefteqn{W_x(E,L,3\ell,s) =}\\
&&\left\{\omega; \;   
\mathrm{dist}\left({\sigma}(H_{{\omega},u,3\ell}) ,E\right) 
> \frac 1 {L^s} \; \; \mbox{for some 
$u \in  \Xi_{L,\ell}(x)$}\right\} \nonumber\
\end{eqnarray}
and
\begin{eqnarray} \label{W1}
W_x(E,L,s)= \left\{\omega; \;   
\mathrm{dist}\left({\sigma}(H_{{\omega},x,L}) ,E\right) 
> \frac 1 {L^s} \right\}\, ,
\end{eqnarray}
We will require $\omega \notin 
W_x(E,L,3\ell,s)\cup  W_x(E,L,s)$, so \eqref{w1} and \eqref{w2} hold.
This will be permissible since it follows from
 \eqref{wegner} that
\begin{eqnarray} \nonumber
\mathbb{P} \{W_x(E,L,3\ell,s)\cup  W_x(E,L,s) \}
&\le&  (3 \frac L \ell+1)^d  Q_I \frac{(3\ell)^{bd}}{ L^{s}} +  
 Q_I  \frac 1 {L^{s -bd}}\\
&<&  \frac{1}{2\ell^{\alpha {p}}}
\le  \frac{1}{2L^{ {p}}}  \label{probw}
\end{eqnarray}
for large $\ell$, since we chose $s > {p} + bd$.

Thus if  $\omega \notin Q_{x}^{(2)}(E,\ell,L,m) 
\cup W_x(E,L,3\ell,s)\cup  W_x(E,L,s)$,
  for each ${x_0\in\Xi_{\frac{L +\ell}3,\ell}}(x)$ we find that
after applying either
\eqref{sligood} or \eqref{slibad2} with \eqref{w1} repeatedly,  stopping
before we 
 hit the boundary belt $\tilde{\Upsilon}_{L}(x)$, we have
\begin{eqnarray} \label{sl360}
\lefteqn{ \| {\Gamma_{x,L}} R_{\omega,x,L}(E) \chi_{x_0,{\frac \ell 3}}\| }\\  && \le
\left(3^d \gamma_{I}   \mathrm{e}^{-m{\ell\over  2}}\right)^{N_r} 
\left(27^d \gamma_{I}^2 L^s  \mathrm{e}^{-m{\ell\over  2}}\right)^{N_s}
 \| {\Gamma_{x,L}} R_{\omega,x,L}(E) \chi_{x_N,{\frac \ell 3}}\|  \nonumber \\
 && \le
\left(3^d \gamma_{I}   \mathrm{e}^{-m{\ell\over  2}}\right)^{N_r} 
\left(27^d \gamma_{I}^2 L^s  \mathrm{e}^{-m{\ell\over  2}}\right)^{N_s}
L^s  \nonumber \, ,
\end{eqnarray}
where $N_r$ and $N_s$ are the number of times we used 
\eqref{sligood} or \eqref{slibad2}, respectively, $N= N_r + N_s$. 
Since $ m \ge  2\theta^\prime {\log \ell \over\ell}$
and $\theta^\prime > \alpha s$,
we can take $\ell$ sufficiently large such that
\begin{equation}  \label{1/2}
27^d \gamma_{I}^2 L^s  \mathrm{e}^{-m{\ell\over  2}} \le 
27^d \gamma_{I}^2 \frac 1 {\ell^{\theta^\prime -\alpha s}}     <  \frac 1 2 \, .
\end{equation}
Combining \eqref{sl360} and \eqref{1/2}, we get
\begin{eqnarray} \label{sl366}
 \| {\Gamma_{x,L}} R_{\omega,x,L}(E) \chi_{x_0,{\frac \ell 3}}\| 
\le
\left(3^d \gamma_{I}   \mathrm{e}^{-m{\ell\over  2}}\right)^{N_r} 
\frac 1 {2^{N_s}}  L^s   \, .
\end{eqnarray}
We cannot  hit the boundary belt $\tilde{\Upsilon}_{L}(x)$ as long
\begin{equation}
(N_r-1)\frac \ell 3 \le \frac{L -3 } 2 - \frac \ell 2 - \frac {L+\ell} 6 -
\frac {8 \ell}3 \, ,
\end{equation}
where we subtracted $\frac {8 \ell}3$ due to the fact that we may
 have gone through the bad region.  Thus we  always have
\eqref{sl366} if
\begin{equation}
N_r \le \frac L \ell -10 \, .
\end{equation}
We have then two possible cases:
 either 
 $N_s$  is large enough so that the right hand side of \eqref{sl366} is 
$\le  \mathrm{e}^{-m{L\over  2}}L^s$, or we get \eqref{sl366}
with $N_r$ the integer satisfying
\begin{equation}
 \frac L \ell -11 <N_r \le \frac L \ell -10 \, ,
\end{equation}and hence
\begin{eqnarray} \label{sl367}
 \| {\Gamma_{x,L}} R_{\omega,x,L}(E) \chi_{x_0,{\frac \ell 3}}\| 
\le
\left(3^d \gamma_{I}   \mathrm{e}^{-m{\ell\over  2}}\right)^{\frac L \ell -11} 
  L^s   \, .
\end{eqnarray}
The estimate \eqref{sl367} holds in either case, so we can proceed as in
\eqref{bigboxest} to get
\begin{eqnarray}
\| {\Gamma_{x,L}} R_{\omega,x,L}(E) \chi_{x,\frac L 3}\|   \le  L^{s}
\left({L\over \ell}+2\right)^d  
\left(3^d\gamma_{I}\mathrm{e}^{-m{\ell\over  2}}\right)^{ \frac L \ell -11}
 \equiv \mathrm{e}^{-M{L\over  2}} \, ,
\label{bigboxest39}
\end{eqnarray}
with\begin{eqnarray}
M \ge m\left(1 - \frac {c_1} {\log \ell}\right) \ge   2\theta^\prime {\log L \over L}
\end{eqnarray}
for $\ell$ sufficiently large, with $c_1$ a constant depending only on 
$d, \gamma_I,\theta^\prime,s,\alpha$ and $L_0$.
The desired estimate \eqref{summ}  follows if  $L_0$ is large enough.
Moreover, it follows from \eqref{p999}    and \eqref{probw} that
for sufficently large $L_0$ we have
\begin{equation}
p_L(M)\le \mathbb{P}\{ Q_{x}^{(2)}(E,\ell,L,m) 
\cup W_x(E,L,3\ell,s)\cup  W_x(E,L,s)\} <   \frac{1}{L^{ {p}}} \, .
\end{equation}
The single energy multiscale analysis (\ref{p111th}) is proven.

We now turn to the proof of the energy interval multiscale analysis
\eqref{msres-bstp2}.   We fix   a
compact subinterval ${\tilde{I}_0}$ of $\mathcal{I}$ with
$I_0 \subset \tilde{I}_0^\circ $, so
 $\mathrm{dist}(I_0,\mathcal{I}\backslash{\tilde{I}_0})>0$.  We require 
\eqref{thetap}, \eqref{thetap2}, and
\begin{equation} \label{theta2p}
\theta > 2p +(b+1)d  \, .
\end{equation}
 As before, the proof proceeds by induction.  
 The initial step in the
 induction is given by  \eqref{start1}.
 Given a scale
$L \in6\mathbb{N}$ and  $m>0$, we set 
\begin{eqnarray}  \label{pL15}
P_L(m,x,y) =
\mathbb{P}\,\left\{R\left(m, L,
I,x,y\right)^{\mathrm{c}}\right\}  \, ,
\end{eqnarray}
where  $A^{\mathrm{c}}$ denotes the complement of the event $A$. 
The induction step
 goes from scale $\ell \ge L_0$ to scale 
$L= [\ell^\alpha]_{6\mathbb{N}}\,$: given that 
\begin{equation}
P_\ell(m,x,y) < \frac1{\ell^{2{p}}}\;\;\mbox{for all
  $x, y \in \mathbb{Z}^d$ with 
$|x-y| > \ell +\varrho\,$,  with
$ m=m_\ell \ge  2\theta^\prime {\log \ell \over\ell}$},
\end{equation}
 we prove
\begin{equation}
P_L(M,x,y) <\frac1{L^{2{p}}}\;\mbox{for all
  $x, y \in \mathbb{Z}^d$ with 
$|x-y| > \ell +\varrho$, some
$M= m _L \ge  2\theta^\prime {\log L \over L}$}.
\end{equation}
To finish the proof of
 (\ref{msres-bstp2}), we   show that that \eqref {summ}
holds for these $m_{L_k}$'s.

The deterministic part of the argument is  quite similar to the one
 we used for the single energy multiscale analysis, except that the probabilistic
estimates will require us to accept the possibility of more singular boxes; 
for every $E \in I$ we will
 forbid the existence of  \emph{four} nonoverlapping  singular boxes in 
either ${\mathcal{C}}_{L,\ell}(x)$ or ${\mathcal{C}}_{L,\ell}(y)$. 
 But the probabilistic estimates
 will require some new ideas.

 Let  ${x} \in \mathbb{Z}^d$  
 and $E \in I$,
and suppose there are at most three nonoverlapping  
 $(\omega,{m}, E)$-singular boxes in 
${\mathcal{C}}_{L,\ell}({x})$, i.e., $ \omega \notin Q_{x}^{(4)}(E,\ell,L,m)$.
  In this case we can always find three boxes 
$\Lambda_\ell(u_i) \in {\mathcal{C}}_{L,\ell}({x})$, $i=1,2,3$,
with $|u_i -u_j| > \ell +\varrho$ if $i\not= j$, such that
 to guarantee
 that $\Lambda_\ell(u^\prime) \in {\mathcal{C}}_{L,\ell}({x})$ is
a $(\omega,{m}, E)$-regular box we need $|u ^\prime - u_i|  > \ell + \varrho$ for each $i=1,2,3$.  
(Note that the $u_i$ depend on
 $\omega, \ell,m,E$.  We may not need all three boxes, but under 
our hypothesis it is always true with three.)
Taking $\ell >3\varrho$,
it suffices to have  $|u ^\prime - u_i|  > \frac {4\ell}3 $ for all $i=1,2,3$. 
We have three cases:
\begin{enumerate}
\item The  closed  boxes $\bar{\Lambda}_{\frac {7\ell}3}(u_i)$,  $i=1,2,3$, are all disjoint.
In this case they are the ``singular regions".

\item Two of the closed   boxes  $\bar{\Lambda}_{\frac {7\ell}3}(u_i)$, say $i=1,2$,
 are not disjoint, with the third closed   box disjoint from the others.
In this case we can find $u_{1,2} \in  \Xi_{L,\ell}({x})$ such that
 $\Lambda_{\frac {7\ell}3}(u_3)$ and 
 $\Lambda_{ {5\ell}}(u_{1,2})$ are our ``singular regions".

\item None of the three closed  boxes  $\bar{\Lambda}_{\frac {7\ell}3}(u_i)$,  $i=1,2,3$
is disjoint from the other two. In this case we can find 
$u_{1,2,3} \in  \Xi_{L,\ell}({x})$ such that
$\Lambda_{ {7\ell}}(u_{1,2,3})$ is our ``singular region".
\end{enumerate}
The point is that all  boxes in 
 ${\mathcal{C}}_{L,\ell}({x})$ with cores  outside
the``singular regions" are regular.  
In all three cases we can find $v_{j} \in  \Xi_{L,\ell}({x})$, 
$\ell_j \in \{{\frac {7\ell}3}, 5 \ell, 7\ell\}$, with
$j=1,\ldots,r\le 3$,  $\sum_{j=1}^r \ell_j \le {\frac {22\ell}3}$, such that 
the closed boxes  $  \bar{\Lambda}_{\ell_j}(v_j)$ are  disjoint and
all  boxes in 
 ${\mathcal{C}}_{L,\ell}({x})$ with cores  outside 
$\bigcup_{j=1}^r   \Lambda_{\ell_j}(v_j)$ are  $(\omega,{m}, E)$-regular.

Given ${x_0\in\Xi_{\frac{L +\ell}3,\ell}}({x})$, we estimate
$ \| \Gamma_{x,L} R_{\omega,x,L}(E) \chi_{x_0,{\frac \ell 3}}\| $ as before
 by applying
 the SLI estimate
 (\ref{sli}) repeatedly,  as long as we do not
 hit the boundary belt $\tilde{\Upsilon}_{L}({x})$.  We now
 have the following cases:
\begin{itemize}
\item If  $x^\prime \notin\bigcup_{j=1}^r   \Lambda_{\ell_j}(v_j)$ and 
 $\Lambda_\ell(x^\prime)  \in
{\mathcal{C}}_{L,\ell}({x})$,  then $x^\prime$
 is the center of a regular box in  ${\mathcal{C}}_{L,\ell}({x})$ and
we use  (\ref{sligood}). 

\item
 If   $x^\prime \in \Lambda_{\ell_j}(v_j)$
and $\Lambda_{\ell_j + {\frac {4\ell}3}}(v_j)\sqsubset \Lambda_L({x})$, we apply
the SLI estimate (\ref{sli})   with $y=x^\prime$, $y^\prime= v_j$, and
 $\ell^\prime=\ell_j + {\frac {2\ell}3}$, so $k\le 23$, obtaining
\begin{eqnarray} \label{slisingular44}
\lefteqn{\| \Gamma_{{x},L} R_{\omega,{x},L}(E) \chi_{x^\prime,{\frac \ell 3}}\| }\\
&& \le
23^d\gamma_{I}   
\| \Gamma_{v_j,\ell_j + {\frac {2\ell}3}}
 R_{\omega,v_j,\ell_j + {\frac {2\ell}3}}(E) \chi_{x^\prime,{\frac \ell 3}}\| 
\| \Gamma_{{x},L} R_{\omega,{x},L}(E)  \chi_{x^{\prime\prime},{{\frac \ell 3}}}\|
\nonumber
\end{eqnarray} 
for some $x^{\prime\prime} \in \Upsilon_{\ell_j + {\frac {2\ell}3},\ell}(v_j)$
(see \eqref{2Upsilon}), so 
$|x^{\prime\prime}-v_j| =\frac {\ell_j }2+ {\frac {\ell}6}$, and hence
$x^{\prime\prime} \notin
  \bigcup_{{j^\prime}=1}^r   \Lambda_{\ell_{j^\prime}}(v_{j^\prime})$ with
$\Lambda_\ell(x^{\prime\prime}) \in {\mathcal{C}}_{L,\ell}({x})$.
We are now in the previous case, so we can use   (\ref{sligood}) to get
\begin{eqnarray} \label{slibad24}
\lefteqn{\|  \Gamma_{{x},L} R_{\omega,{x},L}(E)  \chi_{x^\prime,{\frac \ell 3}}\| }\\
&& \le 
69^d \gamma_{I}^2   \mathrm{e}^{-m{\ell\over  2}}
\|  R_{\omega,v_j,\ell_j + {\frac {2\ell}3}}(E) \| 
\| \Gamma_{{x},L} R_{\omega,{x},L}(E) \chi_{x^{\prime\prime\prime},{{\frac \ell 3}}}\|
\nonumber
\end{eqnarray} 
for some
 $x^{\prime\prime\prime} \in \Upsilon_{\ell,\ell}(x^{\prime\prime})
$; note $|x^{\prime\prime\prime}-v_j| \le \frac {\ell_j + \ell} 2$  
and $|x^{\prime\prime\prime}-x^\prime| \le \ell_j + \frac {\ell} 3$.

\end{itemize}

To control $\|  R_{\omega,v_j,\ell_j + {\frac {2\ell}3}}(E)  \| $ in
 \eqref{slibad24}
we now require  
\begin{eqnarray}  \label{w14}
\|  R_{\omega,v,\ell^\prime}(E) \| \le L^s \;\; \mbox{for all 
$v \in  \Xi_{L,\ell}({x})$ and
 $\ell^\prime \in \{3 \ell, \frac {17\ell}3,\frac {23\ell}3\}$}\, ,
\end{eqnarray}
i.e.,
\begin{equation}
\omega \notin \
\bigcup_{\ell^\prime \in\{3 \ell, \frac {17\ell}3,\frac {23\ell}3\}}
W_x(E,L,\ell^\prime,s)\, . 
\end{equation}
Given ${x_0\in\Xi_{\frac{L +\ell}3,\ell}}({x})$,  we  apply either
\eqref{sligood} or \eqref{slibad24} with \eqref{w14} repeatedly,
 as long as we do not
 hit the boundary belt $\tilde{\Upsilon}_{L}(x)$, obtaining
\begin{eqnarray} \label{sl3604}
\lefteqn{ \|  \Gamma_{{x},L} R_{\omega,{x},L}(E) \chi_{x_0,{\frac \ell 3}}\| }\\  && \le
\left(3^d \gamma_{I}   \mathrm{e}^{-m{\ell\over  2}}\right)^{{N_r}} 
\left(69^d \gamma_{I}^2 L^s  \mathrm{e}^{-m{\ell\over  2}}\right)^{N_s}
 \|  \Gamma_{{x},L} R_{\omega,{x},L}(E)  \chi_{x_N,{\frac \ell 3}}\|  \nonumber \, ,
\end{eqnarray}
where ${N_r}$ and $N_s$ are the number of times we used 
\eqref{sligood} or \eqref{slibad24} with \eqref{w14}, respectively and
 $N+N_r +N_s$.
Since $ m \ge  2\theta^\prime {\log \ell \over\ell}$
and $\theta^\prime > \alpha s$,
we can take $\ell$ sufficiently large such that
\begin{equation}  \label{1/24}
69^d \gamma_{I}^2 L^s  \mathrm{e}^{-m{\ell\over  2}} \le 
69^d \gamma_{I}^2 \frac 1 {\ell^{\theta^\prime -\alpha s}}     <  \frac 1 2 \, .
\end{equation}
Combining \eqref{sl3604}, \eqref{1/24}, and  taking
$\omega \notin W_x(E,L,s)$, i.e., 
$\|  R_{\omega,x,L}(E) \| \le L^s$,
 we get
\eqref{sl366}, but now to guarantee that 
we do not  hit the boundary belt $\tilde{\Upsilon}_{L}(x)$ we need
\begin{equation}
({N_r}-1)\frac \ell 3 \le \frac{L -3 } 2 - \frac \ell 2 - \frac {L+\ell} 6 -
8 \ell \, ,
\end{equation}
where we subtracted ${8 \ell}$ due to the fact that we may
 have gone through the bad regions.  Thus we  always have
\eqref{sl366} if
\begin{equation}
{N_r} \le \frac L \ell -26 \, .
\end{equation}
As before, we have  two possibilities:
 either 
 $N_s$  is large enough so that the right hand side of \eqref{sl366} is 
$\le  \mathrm{e}^{-m{L\over  2}}L^s$, or we get \eqref{sl366}
with ${N_r}$ the integer satisfying
\begin{equation}
 \frac L \ell -27 <{N_r} \le \frac L \ell -26 \, ,
\end{equation}
and hence
\begin{eqnarray} \label{sl3674}
 \| \Gamma_{{x},L} R_{\omega,{x},L}(E)  \chi_{x_0,{\frac \ell 3}}\| 
\le
\left(3^d \gamma_{I}   \mathrm{e}^{-m{\ell\over  2}}\right)^{\frac L \ell -27} 
  L^s   \, .
\end{eqnarray}
The estimate \eqref{sl3674} holds in either case, so we can proceed as in
\eqref{bigboxest} to get
\begin{eqnarray}\label{decay7}
\|  \Gamma_{{x},L} R_{\omega,{x},L}(E)  \chi_{x_0,\frac L 3}\|   \le  L^{s}
\left({L\over \ell}+2\right)^d  
\left(3^d\gamma_{I}\mathrm{e}^{-m{\ell\over  2}}\right)^{ \frac L \ell -27}
 \equiv \mathrm{e}^{-M{L\over  2}}
\label{bigboxest37}
\end{eqnarray}
with\begin{eqnarray} \label{M7}
M \ge m\left(1 - \frac {c_2} {\log \ell}\right) \ge   2\theta^\prime {\log L \over L}
\end{eqnarray}
for $\ell$ sufficiently large, with $c_2$ a constant depending only on 
$d, \gamma_I,\theta^\prime,s,\alpha$ and $L_0$.
The desired estimate \eqref{summ}  follows if  $L_0$ is large enough.

To finish the proof we need to establish the desired estimate on $P_L (M,x,y)$,
where 
  $x, y \in \mathbb{Z}^d$ with 
$|x-y| > L +\varrho$.  Given $u \in \mathbb{Z}^d$, let $Q^{(K)}_{u}(I,\ell,L,m)$ be
 the event that there is an energy $E \in I$
for which  ${\mathcal{C}}_{L,\ell}(u)$
contains  at least $K$ $(\omega,{m}, E)$-singular nonoverlapping  boxes,
i.e.,
\begin{eqnarray}
Q^{(K)}_{u}(I,\ell,L,m)= \bigcup_{E \in I} Q_{u}^{(K)}(E,\ell,L,m)\, ,
\end{eqnarray}  
and let
\begin{eqnarray}
V_{u}(I,\ell,L,s) = \bigcup_{E \in I}
 \left[\left(\bigcup_{\ell^\prime \in\{3 \ell, \frac {17\ell}3,\frac {23\ell}3\}}
W_u (E,L,\ell^\prime,s)\right)\cup W_u (E,L,s)\right]\, .
\end{eqnarray}
We set 
\begin{equation}
Q^{(K)}_{x,y}(I,\ell,L,m)=
 Q^{(K)}_{x}(I,\ell,L,m)\cup Q^{(K)}_{y}(I,\ell,L,m) \, ,
\end{equation}
and
\begin{equation}
V_{x,y}(I,\ell,L,s) =  V_{x}(I,\ell,L,s) \cap V_{y}(I,\ell,L,s)\, .
\end{equation}
 If
$\omega \notin Q^{(4)}_{x,y}(I,\ell,L,m)
\cup V_{x,y}(I,\ell,L,s)$, 
 for every $E\in I$ we have \eqref{decay7}
 and \eqref{M7} for either $ \Lambda_L(x)$ or $\Lambda_L (y)$,
and hence, using the tranlstion invariance of the probabilities,
we have
\begin{equation} \label{probest}
P_L(M,x,y) \le 2\mathbb{P}\{Q^{(4)}_{0}(I,\ell,L,m)\} + 
 \mathbb{P}\{V_{x,y}(I,\ell,L,s)\} \, .
\end{equation}

We first estimate $ \mathbb{P}\{Q^{(4)}_{0}(I,\ell,L,m)\}$.  
Let  ${\mathcal{C}}_{L,\ell}^{(K)}$ denote be the collection
of  $K$ nonoverlapping boxes in ${\mathcal{C}}_{L,\ell}$. 
We have, using  property (IAD) and  the induction hypothesis, that
\begin{eqnarray} \label{PQ2}
\lefteqn{\mathbb{P}\{Q^{(4)}_{0}(I,\ell,L,m)\}   } \\
&& \le 
\sum_{\{\Lambda_\ell(u), \Lambda_\ell(v)\} \in {\mathcal{C}}_{L,\ell}^{(2)}}
\hspace{-.2in} \mathbb{P}\{R(m,\ell,I,u,v)^{\mathrm{c}}\} \,
\mathbb{P}\left\{\bigcup_{\stackrel
{\{\Lambda_\ell(u^\prime), \Lambda_\ell(v^\prime)\}
 \in {\mathcal{C}}_{L,\ell}^{(2)}}
{\{\Lambda_\ell(u), \Lambda_\ell(v), \Lambda_\ell(u^\prime),
 \Lambda_\ell(v^\prime)\} \in {\mathcal{C}}_{L,\ell}^{(4)}}} \hspace{-.5in}
R(m,\ell,I,u^\prime,v^\prime)^{\mathrm{c}}\right\}  \nonumber \\
&& \le 
\sum_{\{\Lambda_\ell(u), \Lambda_\ell(v)\} \in {\mathcal{C}}_{L,\ell}^{(2)}}
\hspace{-.2in} \mathbb{P}\{R(m,\ell,I,u,v)^{\mathrm{c}}\} \,
\mathbb{P}\left\{\bigcup_{
{\{\Lambda_\ell(u^\prime), \Lambda_\ell(v^\prime)\}
 \in {\mathcal{C}}_{L,\ell}^{(2)}}} \hspace{-0.2in}
R(m,\ell,I,u^\prime,v^\prime)^{\mathrm{c}}\right\}  \nonumber \\
&&  \le  \left(\sum_{\{\Lambda_\ell(u), \Lambda_\ell(v)\}
 \in {\mathcal{C}}_{L,\ell}^{(2)}}
\hspace{-.2in} \mathbb{P}\{R(m,\ell,I,u,v)^{\mathrm{c}}\} 
\right)^2   =  \left(\sum_{\{\Lambda_\ell(u), \Lambda_\ell(v)\}
 \in {\mathcal{C}}_{L,\ell}^{(2)}}
\hspace{-.2in} P_\ell (m,u,v)
\right)^2 \nonumber \\ && 
\le \left( \left( 3 \frac L \ell\right)^{2d} \frac1{\ell^{2{p}}}\right)^2
 \le  3^{4d} \frac1{\ell^{4({p} - d(\alpha -1))}} \, . \nonumber
\end{eqnarray}

It remains to estimate $\mathbb{P}\{V_{x,y}(I,\ell,L,s)\}$.  
 Let 
$\tilde{\sigma}(A)= \sigma(A)\cap {\tilde{I}_0}$ for any operator
 $A$.
If $\Lambda_{\ell_1}(u)$ and $\Lambda_{\ell_2}(v)$ are 
nonoverlapping
boxes, then it follows from properties (IAD), (NE) and (W) that
for $ \eta < \mathrm{dist} ( I_0, \mathcal{I}\backslash \tilde{I}_0)$ we have
\begin{equation} \label{wegnerd}
\mathbb{P} \left\{ \mathrm{dist} \left(\tilde{\sigma}(H_{\omega,u,\ell_1}),
\tilde{\sigma}(H_{\omega,v,\ell_2})\right) \le \eta\right\}
\le C_{{\tilde{I}_0}}  Q_{{\tilde{I}_0}} 
\eta \ell_1^{bd} \ell_2^{d} \ .
\end{equation}
To see that, let $\mathcal{F}_1$ and $\mathcal{F}_2$ be the
 $\sigma$-algebras generated by events based on the boxes 
$\Lambda_{\ell_1}(u)$ and $\Lambda_{\ell_2}(v)$, 
 respectively. We set $ \mathbb{P}_i$ to be the restriction of
the probability measure $\mathbb{P}$ to $ \mathcal{F}_i$, with
$\mathbb{E}_i$ the corresponding expectation and $\omega_i$
the corresponding variable of integration,  $i=1,2$.
Using  the independence given by property (IAD), we have
\begin{eqnarray}\label{wegnerd1}
\lefteqn{\mathbb{P} \left\{ \mathrm{dist} \left(\tilde{\sigma}(H_{\omega,u,\ell_1}),
\tilde{\sigma}(H_{\omega,v,\ell_2})\right) \le \eta\right\} = }\\
&&\mathbb{E}_2\left\{\mathbb{P}_1 
 \left\{ \mathrm{dist} \left(\tilde{\sigma}(H_{{\omega_1},u,\ell_1}),
\tilde{\sigma}(H_{\omega_2,v,\ell_2})\right) \le \eta\right\}\right\}\nonumber
\end{eqnarray}
 For a fixed $\omega_2$ we have 
$\tilde{\sigma}(H_{\omega_2,v,\ell_2}) =
\{\lambda_1, \lambda_2, \ldots,\lambda_{N_{\omega_2}}\}$, where
\begin{equation}\label{NEaveraged}
\mathbb{E}_2 (N_{\omega_2})\le  C_{{\tilde{I}_0}}  \ell_2^d
\end{equation}
 by property (NE).  (Note that
$N_{\omega_2}$ and $\lambda_1, \lambda_2, \ldots,\lambda_{N_{\omega_2}}$ depend on
$\omega_2,v,\ell_2$.)  Using property (W), we get
\begin{eqnarray} \label{wegnerd2}
\mathbb{P}_1 
\lefteqn{ \left\{ \mathrm{dist} \left(\tilde{\sigma}(H_{{\omega_1},u,\ell_1}),
\tilde{\sigma}(H_{\omega_2,v,\ell_2})\right) \le \eta\right\} \le}\\
&& \sum_{j=1}^{N_{\omega_2}}  \mathbb{P}_1 
 \left\{ \mathrm{dist} \left(\tilde{\sigma}(H_{{\omega_1},u,\ell_1}),
\lambda_j)\right) \le \eta\right\}  \le  Q_{{\tilde{I}_0}} 
\eta \ell_1^{bd} N_{\omega_2}\, . \nonumber
\end{eqnarray}
The estimate \eqref{wegnerd}  follows from \eqref{wegnerd1},
\eqref{wegnerd2}, and \eqref{NEaveraged}.

Let  $Z_{x,y}(I,\ell,L,s)$  denote the event that 
\begin{eqnarray}
\mathrm{dist} \left(\tilde{\sigma}(H_{\omega,u,\ell_1}),
\tilde{\sigma}(H_{\omega,v,\ell_2})\right) \le \frac 2 {L^s}
\end{eqnarray}
for either
\begin{itemize}
\item[(i)] $u=x$, $v=y$, and $\ell_1=\ell_2=L$,  or

\item[(ii)] $u=x$,  $\ell_1=L$,  and some 
$v \in  \Xi_{L,\ell}(y)$ and $\ell_2 \in 
\{3 \ell, \frac {17\ell}3,\frac {23\ell}3\}$, or

\item[(iii)] $v=y$,  $\ell_2=L$,  and some 
$u \in  \Xi_{L,\ell}(x)$, and $\ell_2 \in 
\{3 \ell, \frac {17\ell}3,\frac {23\ell}3\}$, or

\item[(iv)] some  $u \in  \Xi_{L,\ell}(x)$,
$v \in  \Xi_{L,\ell}(y)$,  and $\ell_1,\ell_2 \in 
\{3 \ell, \frac {17\ell}3,\frac {23\ell}3\}$.

\end{itemize}
Clearly
\begin{equation}
V_{x,y}(I,\ell,L,s) \subset Z_{x,y}(I,\ell,L,s) \, ,
\end{equation}
and it follows from \eqref{wegnerd}, if $L_0$ is large enough so 
$ \frac 1 {L^s} \le 
\frac 1 {L_0^s}< \mathrm{dist} ( I_0, \mathcal{I}\backslash \tilde{I}_0)$, that 
\begin{eqnarray}  \label{wegnerd5}
\lefteqn{\mathbb{P}\left\{ Z_{x,y}(I,\ell,L,s)\right\}
}\\
&\le&  \frac{2 C_{{\tilde{I}_0}}  Q_{{\tilde{I}_0}}}{L^s} 
\left\{L^{(b+1)d} + 
6 \left(3 \frac L \ell\right)^d L^d \left(\frac {23\ell}3 \right)^{bd}
+ \left(3 \frac L \ell\right)^{2d}\left(\frac {23\ell}3 \right)^{(b+1)d}\right\} 
\nonumber \\
&\le& \nonumber  
 \frac{C_{d,b,\alpha} C_{{\tilde{I}_0}}  Q_{{\tilde{I}_0}}}{L^s} 
\left\{L^{(b+1)d} +   \,L^{\left(2   + \frac{b-1}\alpha\right)d}
\right\} \le 
 \frac{2C_{d,b,\alpha}  C_{{\tilde{I}_0}}  Q_{{\tilde{I}_0}}}{L^{s-(b+1)d}} \, ,
\end{eqnarray}
where $C_{d,b,\alpha} $ is a finite constant depending only on $d,b,\alpha$.

It now follows from \eqref{probest}, \eqref{PQ2}, and \eqref{wegnerd5} that 
\begin{eqnarray}
P_L(M,x,y) \le 2\cdot   3^{4d} \frac1{\ell^{4({p} - d(\alpha -1))}} 
+  \frac{4C_{d,b,\alpha}  C_{{\tilde{I}_0}}  Q_{{\tilde{I}_0}}}{L^{s-(b+1)d}}
<\frac1{L^{2{p}}}
\end{eqnarray}
for sufficiently large $L$, since $\alpha < \frac{2{p}+2d}{{p} + 2d}$
 and $s> 2 {p} + (b+1)d$.

\end{proof}

\begin{acknowledgement} It is a pleasure to thank my collaborators on the
 multiscale analysis: Henrique von Dreifus, Alexander Figotin, and Fran\c cois Germinet.
\end{acknowledgement}


\end{document}